\documentclass[11pt]{amsart}
\usepackage{latexsym,amssymb,amsmath,amsthm}
\usepackage{comment,graphicx,color,morefloats}
\usepackage[section]{placeins}
\newtheorem{thm}{Theorem}
\newtheorem{coro}[thm]{Corollary}
\newtheorem{lemma}[thm]{Lemma}
\newtheorem{propo}[thm]{Proposition}
\theoremstyle{definition}

\definecolor{red1}{rgb}{1,0.9,0.9}
\definecolor{blue1}{rgb}{0.9,0.9,1}
\definecolor{green1}{rgb}{0.9,1,0.9}
\definecolor{yellow1}{rgb}{1,1,0.9}
\definecolor{yellow2}{rgb}{1,1,0.8}

\def\question#1{ \vspace{2mm} \begin{center} \parbox{11.2cm}{{\bf Problem:} #1} \vspace{2mm} \end{center} }

\title{The graph spectrum of barycentric refinements}
\author{Oliver Knill}
\date{August 9, 2015}
\address{ Department of Mathematics \\ Harvard University \\ Cambridge, MA, 02138 }

\begin{document}
\maketitle

\begin{abstract}
Given a finite simple graph $G$, let $G_1$ be its barycentric refinement:
it is the graph in which the vertices are the complete subgraphs of $G$ and in which two 
such subgraphs are connected, if one is contained into the other. 
If $\lambda_0=0 \leq \lambda_1 \leq \lambda_2 \leq \cdots \leq \lambda_n$ are the
eigenvalues of the Laplacian of $G$, define the spectral function 
$F(x) = \lambda_{[n x]}$ on the interval $[0,1]$, where $[r]$ is the floor function giving 
the largest integer smaller or equal than $r$. The graph $G_1$ is known to be homotopic to $G$ with 
Euler characteristic $\chi(G_1)=\chi(G)$ and ${\rm dim}(G_1) \geq {\rm dim}(G)$.
Let $G_{m}$ be the sequence of 
barycentric refinements of $G=G_0$. We prove that for any finite simple graph $G$, the spectral functions 
$F_{G_m}$ of successive refinements converge for $m \to \infty$ uniformly on compact subsets of $(0,1)$ 
and exponentially fast to a universal limiting eigenvalue distribution function $F_d$ which only depends on the clique number
respectively the dimension $d$ of the largest complete subgraph of $G$ and not on the starting graph $G$. 
In the case $d=1$, where we deal with graphs without triangles, the limiting 
distribution is the smooth function $F(x) = 4 \sin^2(\pi x/2)$. This is related to the Julia set of the
quadratic map $T(z) = 4z-z^2$ which has the one dimensional Julia set $[0,4]$ and $F$ satisfies
$T(F(k/n))=F(2k/n)$ as the Laplacians satisfy such a renormalization recursion. 
The spectral density in the $d=1$ case is then the arc-sin distribution which is 
the equilibrium measure on the Julia set. In higher dimensions, where the limiting 
function $F$ still remains unidentified, $F'$ appears to have a discrete or singular 
component. We don't know whether there is an analogue renormalization in $d \geq 2$.
The limiting distribution has relations with the limiting vertex degree distribution and 
so in 2 dimensions with the graph curvature distribution of the refinements $G_m$. 
\end{abstract}

\section{Introduction}

The spectral theory of graphs \cite{Chung97,Mieghem,Brouwer,VerdiereGraphSpectra,Post,BLS}
parallels to a great deal the spectral theory of compact Riemannian manifolds 
\cite{Chavel,Rosenberg,BergerPanorama}. There are areas with a good match like
inverse spectral topics, heat kernel related topics like Hodge theory where
the dimension of the space of harmonic solutions $L f = 0$ for example
is the number of connected components and more generally, harmonic $k$-forms are
related to the k'th cohomology groups in a rather explicit way as applying
the heat flow to $k$ forms converges to harmonic forms. This can be used for example to
construct the Hurewicz homomorphism from $\pi_k(G)$ to the cohomology groups $H^k(G)$ by applying the
heat flow $e^{-iL_k}$ on a $k$-form $f$ supported initially on a $k$-sphere representing an element in 
the homotopy group $\pi_k(G)$. It leads to a harmonic $k$-form representing a homology
class. Discrete McKean-Singer \cite{knillmckeansinger} 
${\rm str}(e^{-t L}) = \chi(G)$ illustrates further, how the graph theory 
parallels the continuum \cite{McKeanSinger}: the super trace of the heat kernel for $t=0$
is the super trace of $1$ which is the definition of the Euler characteristic, while
in the case $t \to \infty$, it counts the alternating sum of the Betti numbers, which by Euler-Poincar\'e
is the same. The cohomology of discretizations of manifolds has since
the beginning of the development of algebraic topology been used to compute the cohomology 
of the manifold. For the convergence of the spectrum of $p$-form Laplacians of a compact
connected Riemannian manifold $M$ of dimension $d \geq 2$,  Mantuano's theorem 
\cite{Mantuano} tells that for any $\epsilon$ discretization $G$ given by a graph, 
the spectra of the Hodge Laplacians $L_p$ are related by 
$c \lambda_{k,p}(G) \leq \lambda_{k,p}(M) \leq C \lambda_{k,p}(G)$, where the constants $c,C$
only depend on dimension $d$, the maximal curvature and radius of injectivity of $M$. Also for eigenfunctions,
the number of nodal regions of the $k$ eigenfunction $f_k$ is bound by a theorem of 
Fiedler \cite{Fiedler1975} by $k$, paralleling the Courant nodal theorem. 
The ground state energy $\lambda_1>0$ is estimated in the same way from below by the Cheeger constant.
The nodal regions on geometric graphs correspond to Chladni figures on compact Riemannian manifolds. 
Also in topology like for Brouwer-Lefshetz \cite{brouwergraph} or Jordan-Brouwer theory \cite{KnillJordan},
the notions translate nicely to the discrete and  barycentric refinements help. 
There are also places, where things are different: the Weyl law relating the growth of the eigenvalues with volume
has no direct discrete analogue because the spectrum of a graph is a finite set.
Similarly, the Minashisundaram-Plejel zeta function $\sum_{\lambda>0} \lambda^{-s}$ of a manifold needs 
analytic continuation in the continuum, while in the graph case, the zeta function is an entire function. 
Already in the case of a circle, where the Dirac version of that zeta function is the Riemann 
zeta function, the zeros of the discrete analog on circular graphs can be analyzed well, unlike in the 
continuum, where it it is the Riemann hypothesis. 
The convergence of the roots \cite{KnillZeta} of the graph zeta functions is a limit when looking at barycentric 
subdivisions of circular graphs, which is also important for 
Jacobi matrices generalizing graph Laplacians \cite{BGH,BBM} which we studied in \cite{Kni95}
in an ergodic setup, where the limiting operators are almost periodic Jacobi matrices
over the von Neumann-Kakutani system. The later is the unique fixed point of the 2:1 integral extension operation
in ergodic theory. The hull of the operators is the compact Abelian group of dyadic integers, where space 
translation is addition by 1 and where the renormalization step is the shift. 
We look now at higher dimensional analogue questions to these one dimensional Jacobi spectral problems. 
There are other analogies between Schr\"odinger operators and geometric Laplacians: while isospectral
deformations are possible in the former case, geometric Laplacians are harder to deform as the isospectral
set of geometries is discrete in general. Having worked with isospectral deformations of Jacobi matrices
in \cite{Kni93a,Kni93b} we looked in \cite{Kni94} at higher dimensions and noted that writing $L=D^2$ can enable
isospectral deformation and break the spectral rigidity. Much later, we realized
that in the Riemannian as well as in the graph case, one can deform the exterior derivative
$d$ \cite{IsospectralDirac,IsospectralDirac2} in an isospectral way. The Riemannian case is a bit more
technical in that story as the deformed exterior derivatives $d_t$ are pseudo differential equations in the
continuum but they satisfy $d_t^2=0$ so that
they deform cohomology. The deformation does not affect the Laplacian $L$, but in the complex, the 
nonlinear deformation becomes asymptotic to the wave evolution. Interestingly, both in the discrete as well as in the
continuum, space expands, with an inflationary start. This is not artificially placed into the model,
but is a basic property if one lets the Dirac operator move freely in its symmetry group.
Graphs can serve as a laboratory to test concepts related to physics. It is reasonable for example to see
the evolutions of the various differential forms under the wave dynamics as manifestations
of different forces. It is a caricature for physics, where nothing has put into the system except
for geometry, which is a Riemannian manifold in the continuum and a graph $G$ in the discrete. 
When answering the question which graph to take, symmetry again could give a hint: take a scale invariant graph. 
While impossible to achieve in a finite discrete structure, taking
a graph which is multiple way refined in a barycentric way comes close. Not only are parts of the graph
homeomorphic to the entire graph, but also the nerve graphs of open covers are homeomorphic to the entire graph. 
This motivates to look at refinements and its spectral properties. \\

\begin{figure}
\scalebox{0.08}{\includegraphics{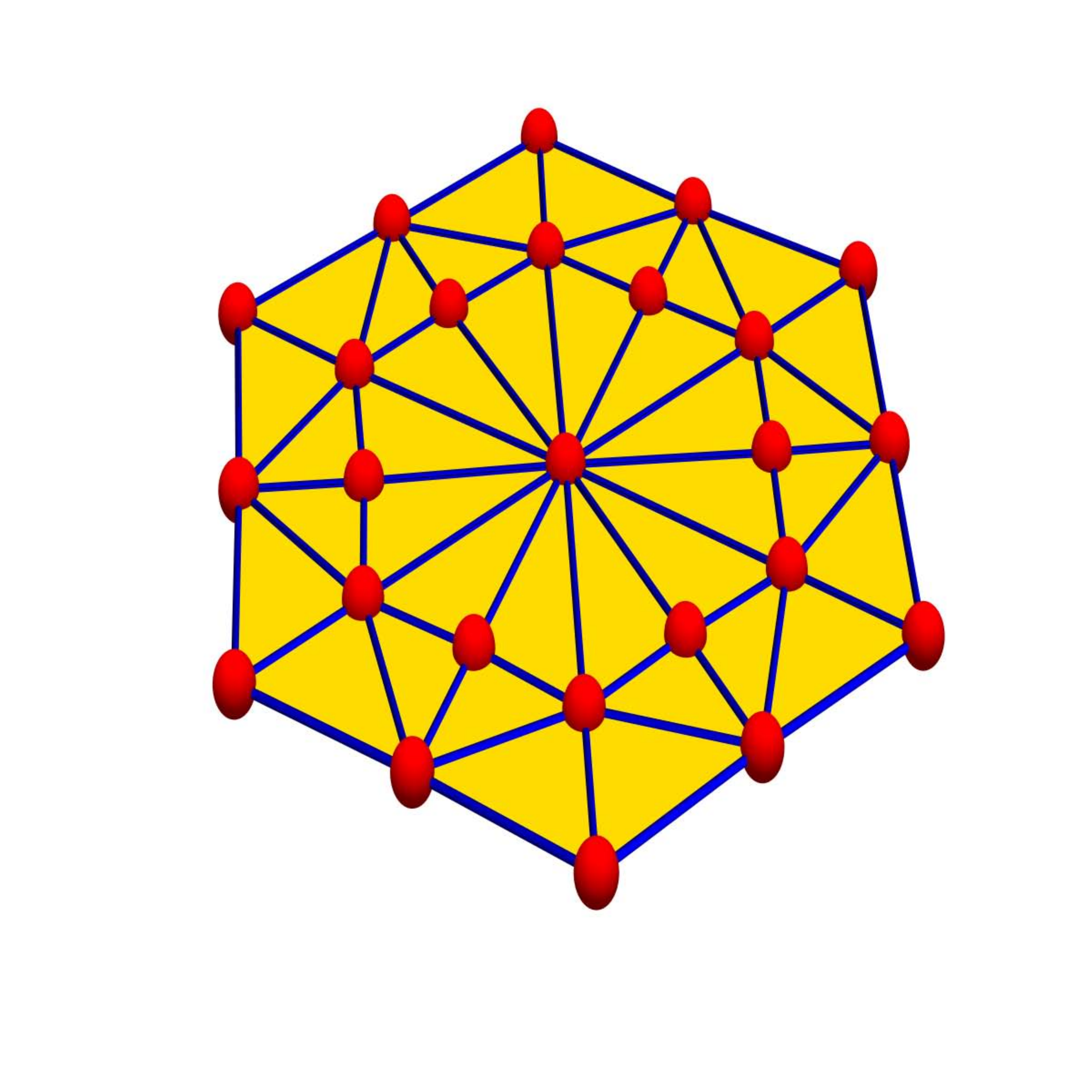}}
\scalebox{0.08}{\includegraphics{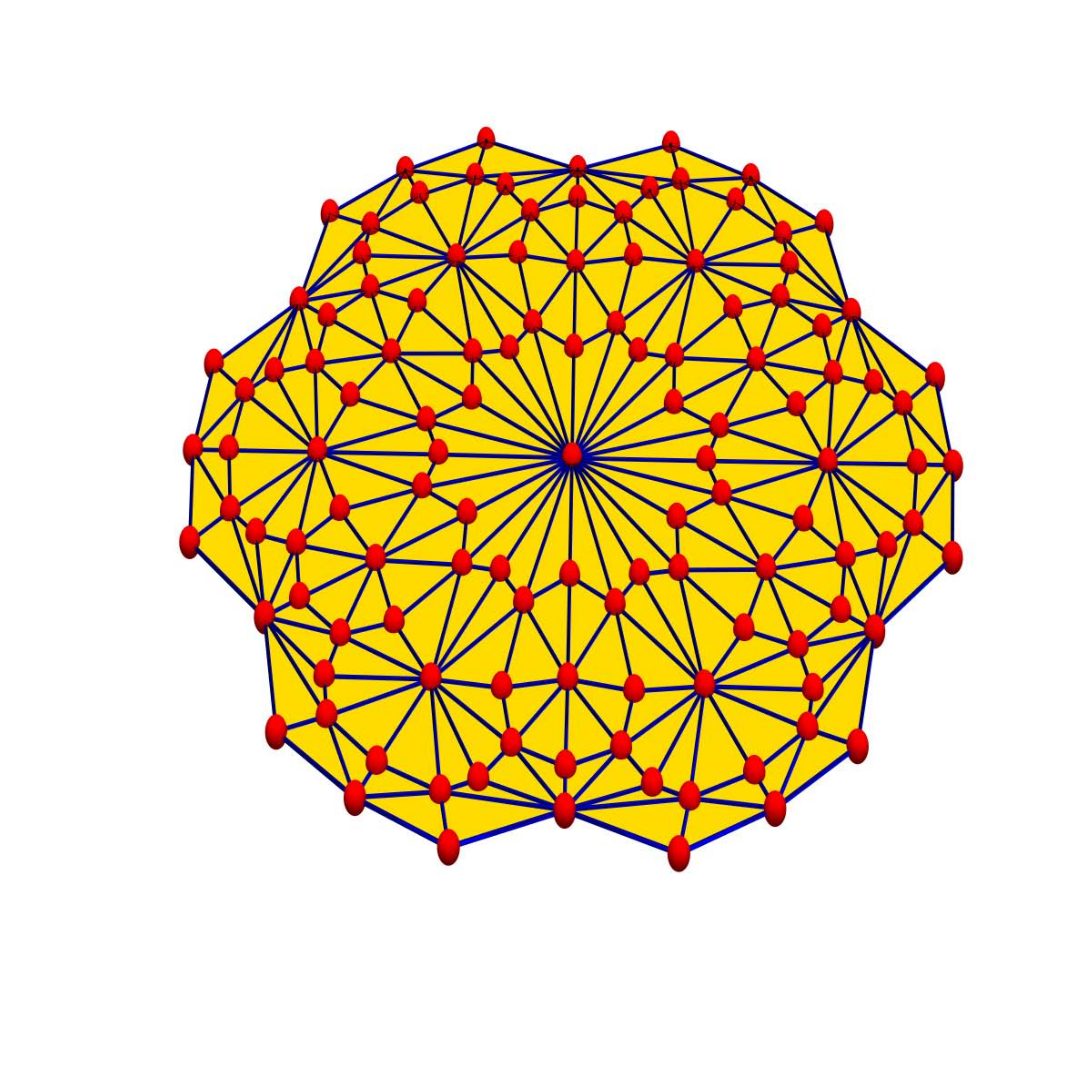}}
\scalebox{0.08}{\includegraphics{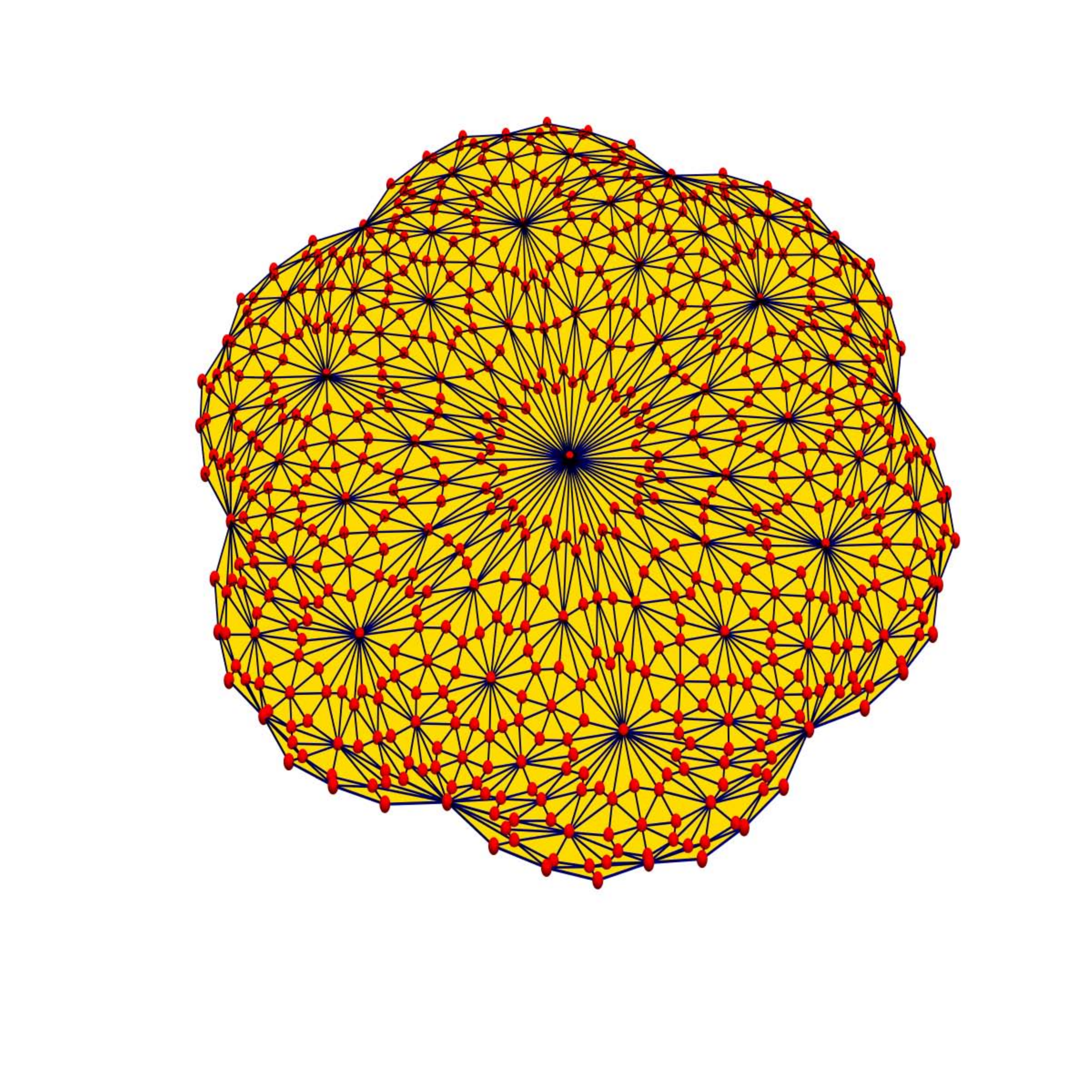}}
\caption{
Barycentric refinements of the triangle $G=K_3$. The number of vertices of the
ball $G_m$ grows exponentially like $O(6^m)$ and is exactly known.
}
\end{figure}

In \cite{KnillProduct}, we looked at a product of graphs which satisfies the same K\"unneth formula
for cohomology as in the continuum. This product of $G$ and $H$ is defined as follows: take two disjoint
unions of $G$ and $H$. Pick a complete subgraph $x$ of $G$ and a complete subgraph $y$ of $H$. 
Let $xy$ be the complete subgraph generated by $x \cup y$. These points $xy$ form the vertices of $G \times H$. 
Two such vertices $xy$ and $uv$ are connected if one is contained in the other. A special case is
if $H=K_1$ in which case $G_1 = G \times K_1$ is the barycentric refinement of $G$: its vertices are the 
complete subgraphs $x$ of $G$ and two complete subgraphs $x,y$ are connected if one is contained in an other. 
The graph $G_1$ is homotopic to $G$. It especially has the same Euler characteristic. The dimension of $G_1$
is bounded below by the dimension of $G$. If $G$ is geometric, then $G_1$ is even homeomorphic to $G$ and the 
dimensions of $G$ and $G_1$ are the same. For any $d$-dimensional geometric $G$, the graphs $G_m$ are all
$(d+1)$-colorable, the color being the dimension of the original simplices which make up now the vertices.
The automorphism group of $G$ also acts on $G_1$ so that
fixed points of graph automorphisms \cite{brouwergraph} can be realized as vertices in $G_1$. 
Barycentric refinements again stress the point of view taken by discrete Morse theory \cite{forman95,Forman1999} 
that complete subgraphs in a graph can be treated as ``points". 

\begin{figure}
\scalebox{0.08}{\includegraphics{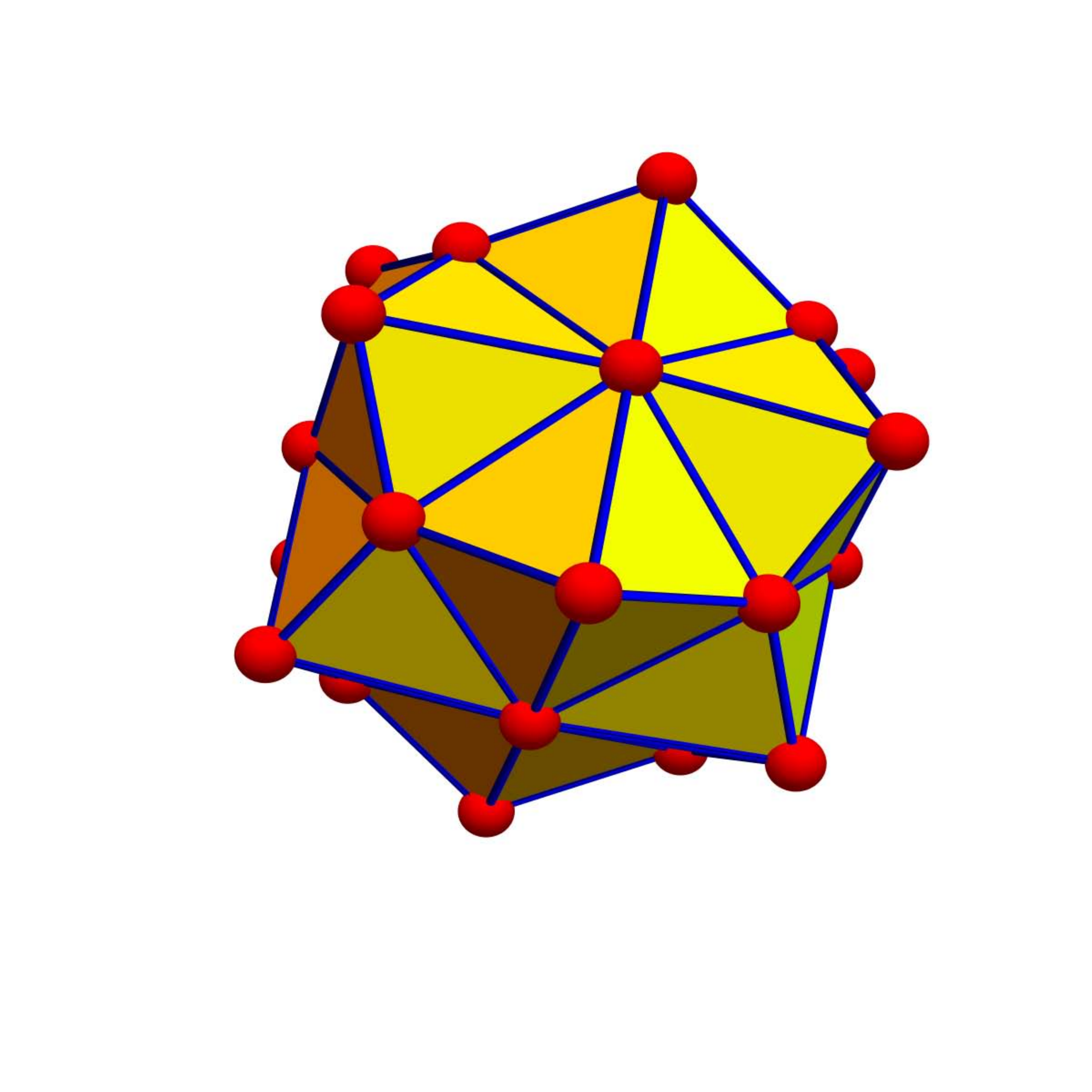}}
\scalebox{0.08}{\includegraphics{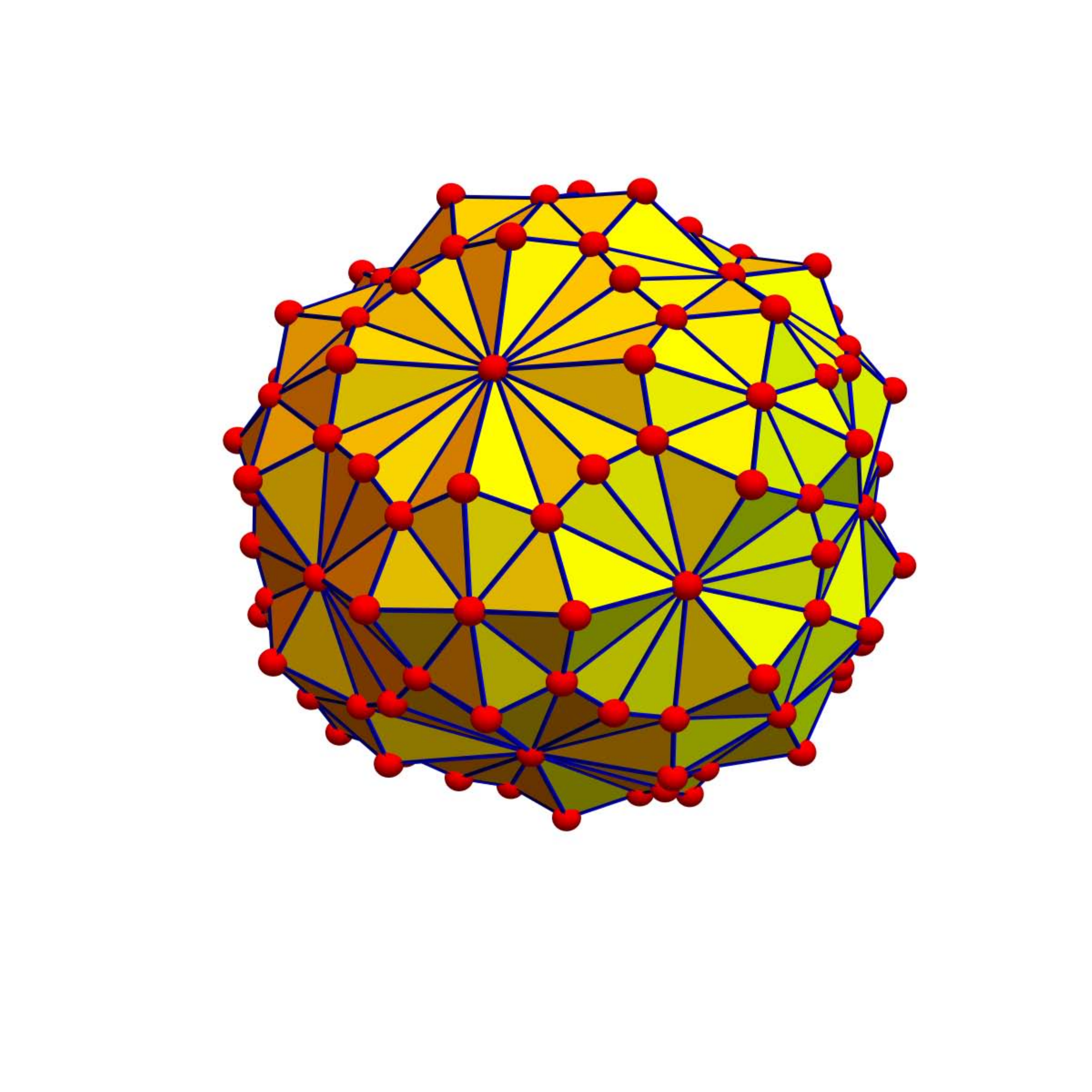}}
\scalebox{0.08}{\includegraphics{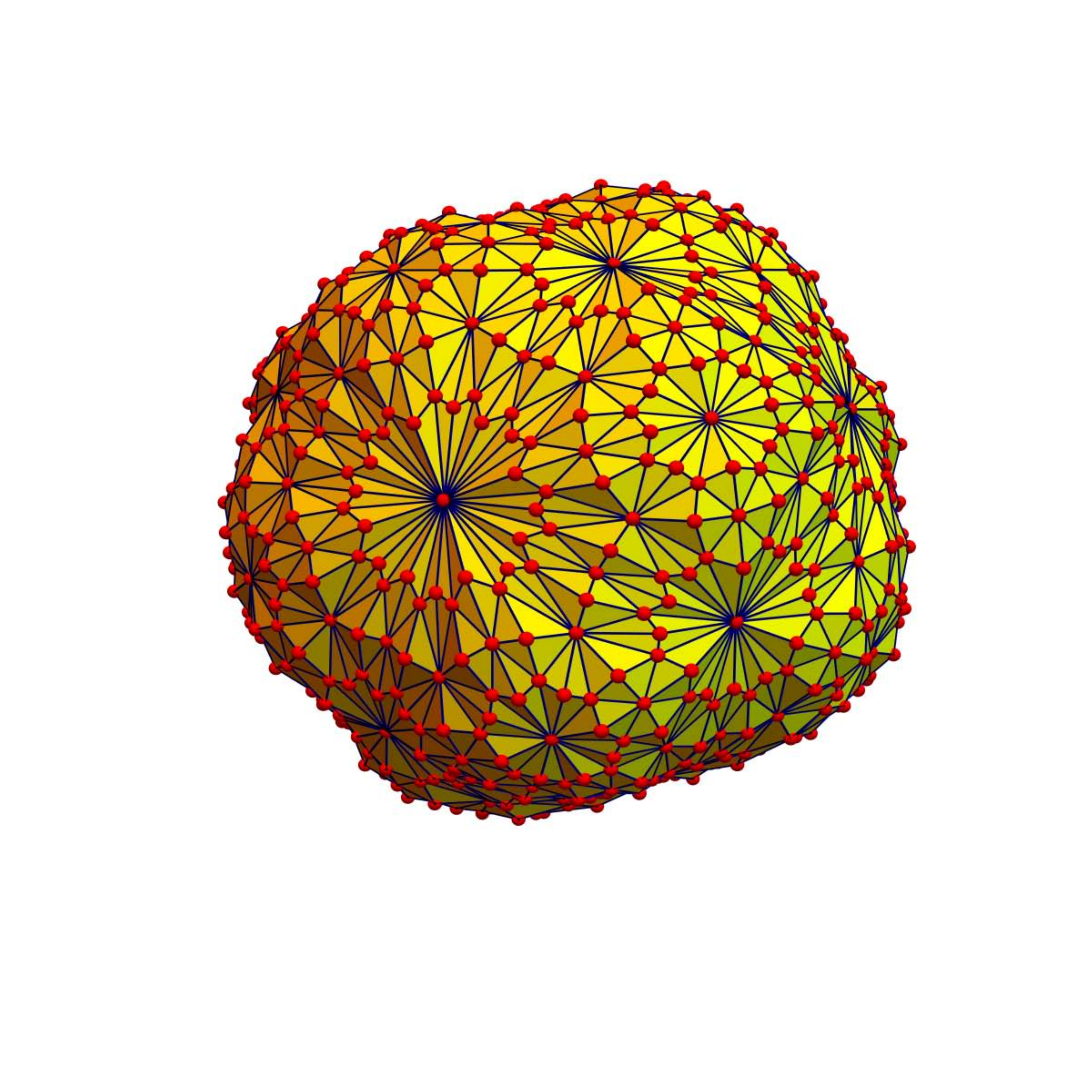}}
\caption{
Barycentric refinements of the octahedron.
}
\end{figure}

For practical triangulations, barycentric refinements are rarely used. The reason is the huge difference
between the possible vertex degrees, as each refinement doubles the maximal degree. But the appearance of 
large vertex degrees has also benefits: for any $m$, the $G_m$ graphs are Eulerian with Eulerian spheres
and have the property that a natural geodesic flow can be defined on them in any dimension. This is essential,
as spheres and lines are important in any geometry and that for defining lines, we need an Eulerian property as with an odd
degree vertex, the continuation of a ``straight line" is ambiguous: we have difficulty for example to continue
a line through a vertex of an icosahedron as we will have to chose what direction to continue. For graphs in 
which all spheres and all spheres in spheres are Eulerian, there is a natural continuation. Graph products and 
especially barycentric refinements have this property. 
The large and nonuniform degree of some density distributed vertices has the advantage as it allows to emulate
the rotational symmetry in the discrete. For a $d$-dimensional graph, and considering $G_m$, 
we can in dimension $d$ start with a vertex degree of $O((d+1)!^{m})$ directions. Unlike in regular
lattices or tessellations, the spheres appear more rounded and the fact that the unit sphere $S(x)$ 
has asymptotically a similar amount of directions than the number of vertices in $G_m$ gives us the property 
that the exponential map $\exp_x$ from the unit sphere to the graph, covers most of the graph.
This is probably the closest we can  get to a Hopf-Rynov property in the discrete without digressing into the
quantum (as nature does) and look at the wave equation $u_{tt} = L u$ for the Hodge Laplacian $L=(d+d^*)^2$,
which because $L=D^2$ has the d'Alembert solution $u(t) = \cos(Dt) u(0) + \sin(Dt) D^{-1} u'(u)$,  a superposition 
of wave group solutions $e^{\pm iDt} = \cos(Dt) \pm i \sin(Dt)$ of the Schroedinger equation $\psi' = \pm i D \psi$ 
with $\psi = u + i D^{-1} u'$ for the Dirac operator $D=d+d^*$ (see \cite{DiracKnill}) for which 
Hopf-Rynov is just linear algebra in the graph case as the unitary group is finite dimensional and
for any two vertices $x,y$ one can solve the problem to start a geodesic on $e_x$ and reach 
$e_y$ by choosing the correct velocity $u'(0)$ and time $t$. \\

\begin{figure}
\scalebox{0.08}{\includegraphics{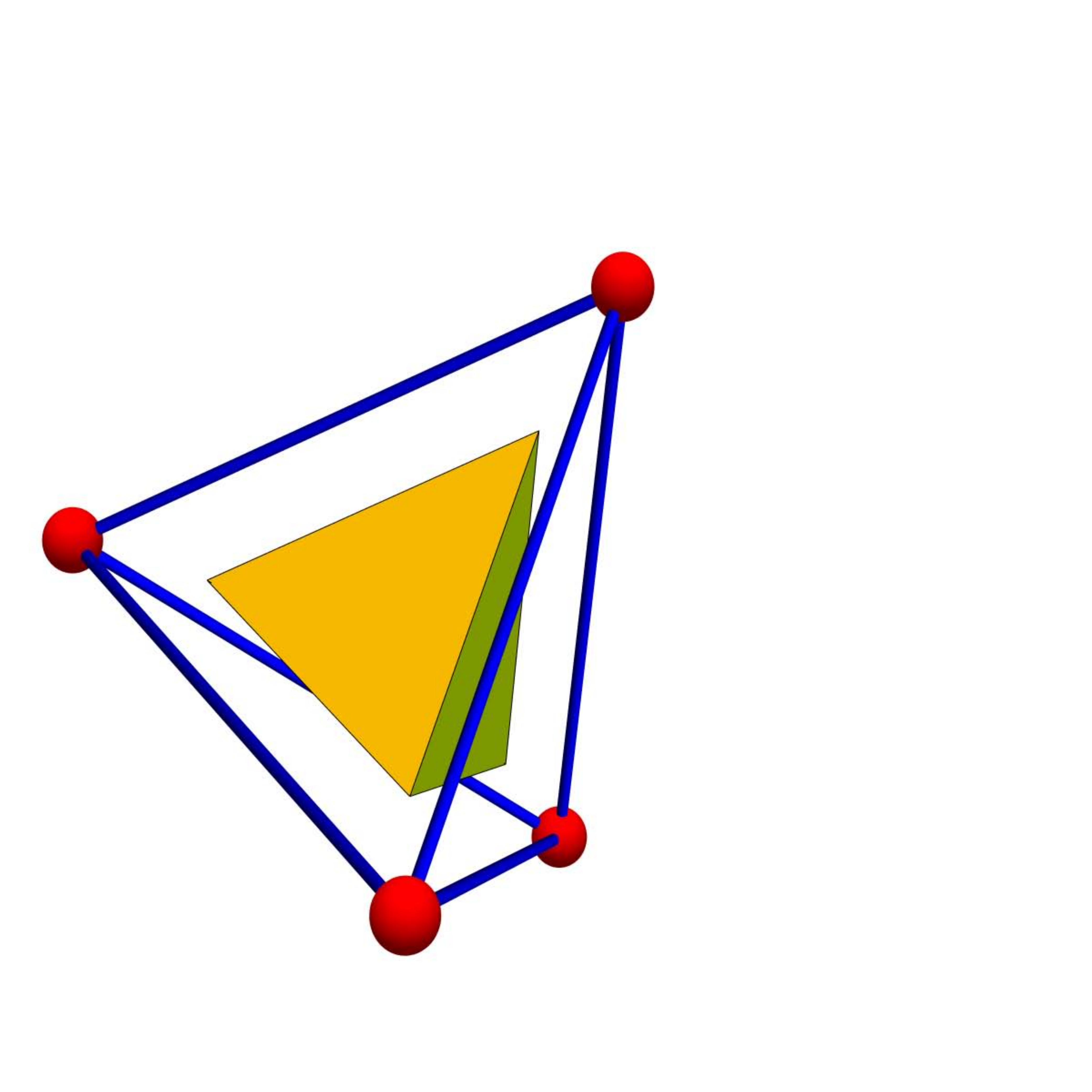}}
\scalebox{0.08}{\includegraphics{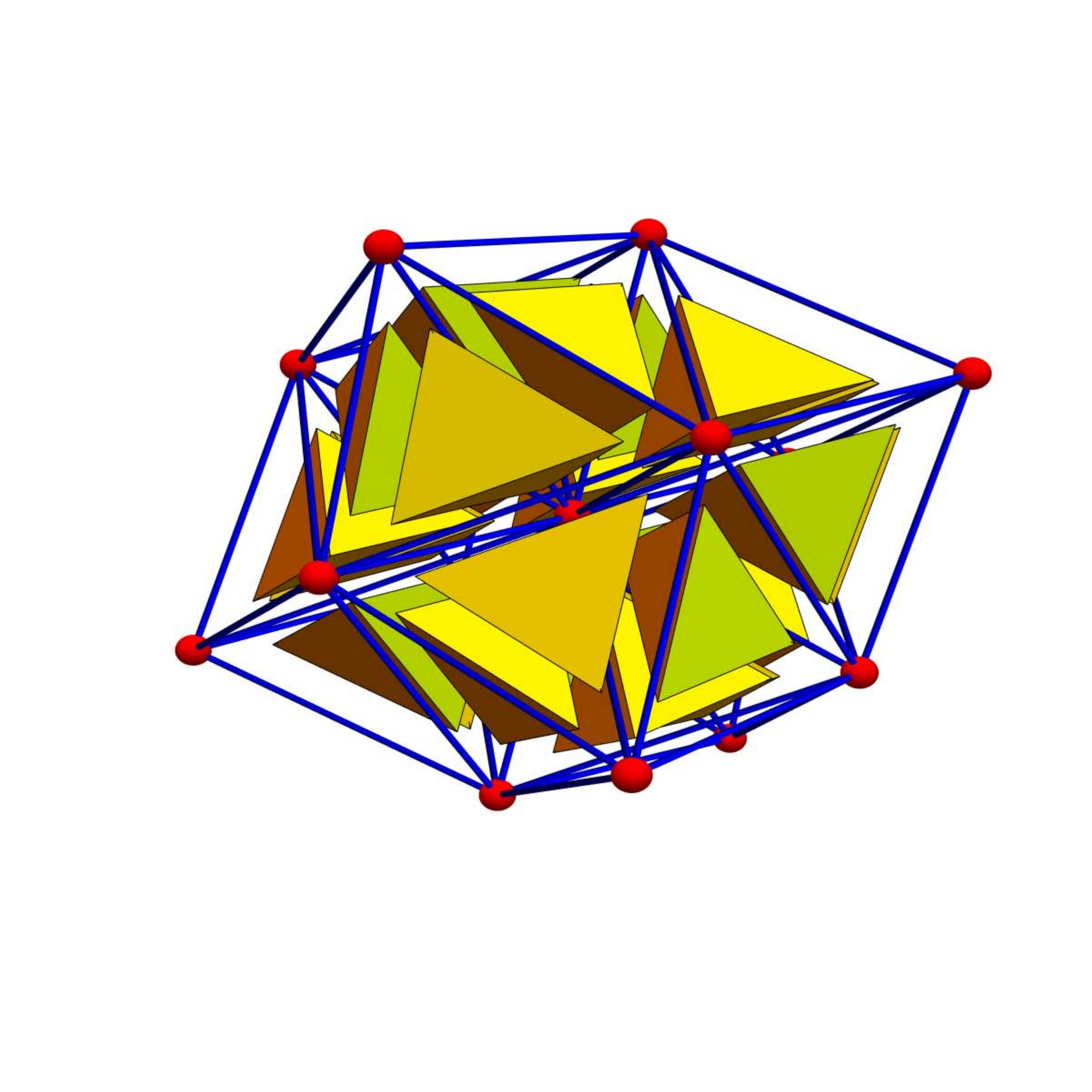}}
\scalebox{0.08}{\includegraphics{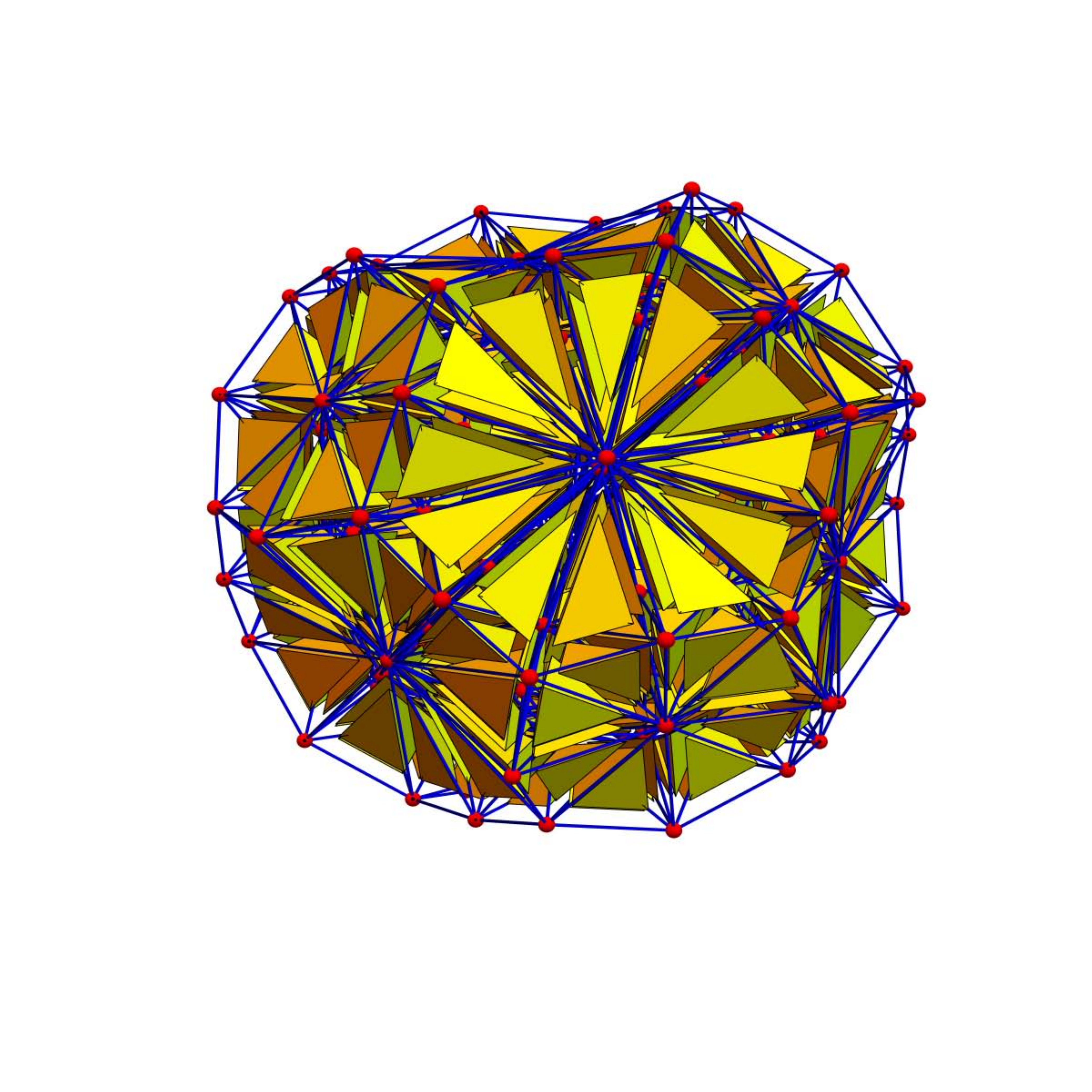}}
\caption{
Barycentric refinements of the tetrahedron $G=K_4$. 
The number of vertices grows like $O(24^m)$.
We know $v_4(G_m)=24^m$ but we don't have an
explicit formula yet for the number of vertices $v_0(G_m)$ if $G_0=K_4$. 
}
\end{figure}

Iterated barycentric subdivisions have long been used in topology and homology theory. 
The notion is mostly defined in geometric realizations used to refine a given triangulation of 
a simplex in Euclidean space.
Note that we don't look at any embedding in an Euclidean space but look at Barycentric 
refinement in a completely combinatorial setup. 
There are also probabilistic connections \cite{DiaconisMiclo,Hough}, 
as iterated barycentric subdivisions of a triangle define a random walk on ${\rm SL}(2,R)$ and
are linked to concrete geometric properties like angles triangles. Such work could suggest that unlike in the
one dimensional case, in higher dimensions, higher dimensional group actions and random walks 
could play a role when studying the spectrum. In graph theory, where no embedding into an Euclidean 
space is required, barycentric subdivisions are known in the context of flag complexes 
or Whitney complexes in particular. We are not aware of any work on the nature of the spectrum of 
such graphs. \\

The inductive dimension of a graph $G$ is defined inductively as the average of the dimensions of all unit 
spheres, incremented by $1$ \cite{elemente11,randomgraph}. 
A graph $G$ is a geometric $d$-dimensional graph, if every unit sphere is a
$(d-1)$-sphere or a $(d-1)$-ball and the boundary set of vertices with a ball unit sphere 
form a $(d-1)$-dimensional geometric graph without boundary. 
Geometric graphs play the role of manifolds with or without boundary. 
The definition of graph theoretical spheres and homotopy are both due to Evako.
We have defined spheres independently in \cite{KnillEulerian,KnillProduct} but realized in 
\cite{KnillJordan} that it is already in the work of Ivashchenko=Evako 
\cite{I94a,I94,Evako1994,Evako2013}. The definition of these Evako spheres is based on Ivashchenko homotopy 
which is homotopy notion inspired by Whitehead \cite{Whitehead} but defined for graphs. 
A simplification done in \cite{CYY} makes this homotopy much more practical, like for Lusternik-Schnirelman
theory for graphs \cite{josellisknill}, to define what a homeomorphism for graph is \cite{KnillTopology} or 
for graph colorings \cite{knillgraphcoloring,KnillEulerian}, where constructing an Eulerian 3 sphere
having a given 2 sphere embedded could explain why 4 colors are enough to color a planar graph. \\

We proved in \cite{KnillProduct} the inequality ${\rm dim}(G \times H) \geq {\rm dim}(G) + {\rm dim}(H)$,
which holds for all finite simple graphs $G,H$. 
This implies in a special case ${\rm dim}(G_1) \geq {\rm dim}G)$, still for all finite simple graphs $G$.
In the geometric case, the dimension of $G_1$ is the same as the dimension of $G$. 
The barycentric refinement process $G \to G_1$ honors geometric graphs, graphs for which the unit spheres are spheres. 
Starting with a $d$-simplex $G_0=K_{d+1}$, successive refinements $G_m$ 
produce geometric graphs which are $d$-balls, for which the boundary is a $(d-1)$-sphere. In the 
case of a triangle $G = K_3$ for example, the graph $G_1$ is a wheel graph with $6$ spikes. It has one interior point and
$6$ boundary point. Successive refinements produce larger and larger balls for which the boundaries are
spheres. \\

If $G$ is a geometric $d$-dimensional graph, then $G_1$ is homeomorphic to $G$ in the sense of
\cite{KnillTopology}: there is an open cover $U_j$ of $G_1$ such that the nerve graph of $G_1$ is $G$. 
To see this, start with unit disks $V_j$ centered at the vertices of $G_1$ which originally were vertices. Then add
to $V_j$ the $d$-simplices on the outside which contain $(d-1)$ simplex in $V_j$ and a vertex simplex $z$. 
This produces $U_j$. Two different $U_j$ intersect in a $d$-simplex, if and only the original vertices were connected.
Homeomorphic graphs have the same topological properties like cohomology, fundamental groups or connectivities. 
We furthermore know that for any starting point $G$, the dimension of $G_m$ converges to the dimension 
of the largest complete subgraph in $G$ \cite{KnillProduct} and the reason is the same as for the spectral 
convergence studied here: the largest dimensional simplices dominate. 

\section{Counting} 

The refinements $G_m$ grow exponentially fast already in the $1$-dimensional case, where the number of 
edges doubles exactly in each step.

\begin{lemma}
a) For $G=K_2$, $v_0(G_m) 1+2^m$ and $v_1(G_m)=2^m$. \\
b) For $G=K_3$, the number of vertices $v_0(G_m) = 1-3(2^{m-1}+2^m+2^{m-1} 3^{m-1})$,
the number of edges is $v_1(G_m) = 3(-2^{m-1}+2^m+2^{m-1} 3^m)$ and the number of triangles
$v_2(G_m)=6^m$.\\
c) For $G=K_d$, we have $v_d(G_m) = ((d+1)!)^m$ and $v_0(G_m) \geq ((d+1)!)^{m-1}$. 
\end{lemma}
\begin{proof}
a) We have $(v_0(0),v_0(1),v_0(2),\dots) = (2,3,5,9,17,33, \dots )$ as $v_0(m+1) = 2 v_0(m)-1$. Furthermore
$v_1(m)+v_0(m)=1$. \\
b) From the affine recurrence given by
$A=\left[ \begin{array}{cccc} 0&2&0&1 \\ 0&2&6&0 \\ 0& 0&6&0 \\ 0&0&0&1 \\ \end{array} \right]$
and $S^{-1} A S= {\rm Diag}(6,2,1,0)$ with 
$S=\left[ \begin{array}{cccc} 1 & 1 & 1 & 1 \\ 3 & 1 & 0 & 0 \\ 2 & 0 & 0 & 0 \\ 0 & 0 & 1 & 0 \\ \end{array} \right]$
we get the exact simplex data formula
$$ S B^m S^{-1} \left[\begin{array}{c}3 \\3\\1\\1 \end{array} \right] 
  = \left[ \begin{array}{c} 1-3 \cdot 2^{m-1}+3 \cdot 2^m+ 2^{m-1} 3^{m}) \\ 
                            3(-2^{m-1}+2^m+2^{m-1} 3^m) \\
                            6^m \\ 1  \end{array} \right] $$
which in the first coordinate gives the above sequence $3,7,25,\dots$ which grows like $6^m/2$. 
$$ (v_0(0),v_0(1),v_0(2),\dots) = (3,7,25,121,673,3937,23425,\dots) \; . $$
c) follows from the fact that in each step, a $d$-dimensional simplex is cut into $(d+1)!$ simplices. 
This shows $v_d(G_m) = (d+1)! v_d(G_{m-1})$. 
Furthermore $v_0(G_m) = \sum_{j} v_j(G_{m-1}) \geq v_d(G_{m-1})$. 
\end{proof}

In \cite{Snyder}, a combinatorial interpretation of the eigenvalues $1,2,6$ are given.
We hope it will be possible to get also exact formulas for $v_k(G_m)$ if $G_0=K_{d+1}$ if $d \geq 3$.
To do so, one could use the handshaking relation (see \cite{cherngaussbonnet})
$$  \sum_{x \in V} V_k(x)=v_{k+1} (k+2)   \; , $$ 
where $V_k(x)$ are the number of $k$-dimensional complete subgraphs $V_{k+1}$ in the unit 
sphere $S(x)$. In the case $k=1$, this is the Euler handshake giving the number $v_1$
of edges in terms of the vertex degrees $V_0(x)$ of the graph. In the graph $G_m$,
the spheres are all joins of smaller-dimensional spheres and in particular come in finitely 
many types. What one would have to do first is to give formulas 
for $v_k(S_{d,l,m}))$ for these $d$-spheres of type $l$ in level $m$ and then get 
$v_k(G_m)$ by the Handshaking lemma using the fact that we know the number of these
spheres. For example, in the case $d=2$, only degrees $4,6,8,12,16,24,..$ can occur in the 
interior and degrees $3,5,9,17, \dots$ at the boundary. It looks like a nice but not hopeless
combinatorial challenge to analyze this. \\

{\bf Some numbers.}\\
1) In the one dimensional case $G=K_2$, where $v_0(G_m)=2^m+1$ and $v_1(G_m)=2^m$, we have 
$$  v_0(G_{m+1})/v_0(G_m) \to 2=2! \; . $$
2) In the two dimensional case $G=K_3$, we have everything explicit. In particular,
$v_0(G_m) = 1-3 2^{m-1}+3 2^m+ 2^{m-1} 3^{m})$ which satisfies 
$$  v_0(G_{m+1})/v_0(G_m) \to 6=3! \; . $$
3) In the three dimensional case $G=K_4$, we have the following data for \\
$\vec{v} = (v_0,v_1,v_2,v_3)$: \\
$\vec{v}(G) = (4, 6, 4, 1)$ \\
$\vec{v}(G_1) = (15, 50, 60, 24)$,  \\
$\vec{v}(G_2) = (149,796,1224,576)$, \\ 
$\vec{v}(G_3) = (2745,17000,28080,13824)$,\\
$\vec{v}(G_4) = (61649,\dots,\dots,331776)$.  \\
We don't have a closed formula for the sequence $v_0(G_m)$ starting with
$4,15,149,2745,\dots$. We only know that $v_0(G_{m+1})/v_0(G_m) \to 24=4!$.  \\
4) In the four dimensional case $G=K_5$, we have computed so far:\\
$\vec{v}(G)=(5,10,10,5,1)$, \\
$\vec{v}(G_1)=(31,180,390,360,120)$, \\
$\vec{v}(G_2)=(1081,11340,33300,37440,14400)$. \\
$\vec{v}(G_3)=(97561,\dots,\dots,\dots,1728000)$. \\
Also here, we only know $v_0(G_{m+1})/v_0(G_m) \to 120=5!$ and have no closed formula
for the sequence $v_0(G_{m})$ starting with $5,31,1081,97561,\dots$ yet.

\section{Spectral distribution functions}

The Laplacian of a graph is the matrix $L=B-A$, where $A$ is the adjacency matrix and
$B$ is the diagonal degree matrix. For a complete graph $K_{d+1}$ for example, the
matrix is everywhere $-1$ except for the diagonal, where the entries are $d$. As
subtracting $d+1$ produces a matrix $B$ with $d$ dimensional kernel and trace $-d-1$,
the eigenvalues of the Laplacian of the complete graph $K_{d+1}$ is $d+1$ with multiplicity
$d$ and $0$ with multiplicity $1$. For the house graph $G$ for example, where a triangle $K_3$
(roof) is glued on top of a square $C_4$ leading to a graph of dimension $22/12$ as it is 
mixture of one and two dimensional components,  the Laplacian is 
$$ L = \left[
                 \begin{array}{ccccc}
                  2 & -1 & 0 & -1 & 0 \\
                  -1 & 3 & -1 & 0 & -1 \\
                  0 & -1 & 3 & -1 & -1 \\
                  -1 & 0 & -1 & 2 & 0 \\
                  0 & -1 & -1 & 0 & 2 \\
                 \end{array}
                 \right]  \; . $$

While experimenting with nodal surfaces for eigenfunctions of the Laplacian on discrete $3$-spheres, where 
we expect the nodal surface of the ground state to be a $2$-sphere in general, we also computed
eigenvalues of refinements  and got interested on how the eigenvalue distribution depends on refinements. 
While we expected some limit to be reached, we would have thought that the limit to depend on the topology
of the initial graph. This is not the case. Starting with a $2$-torus or a 2-sphere produces the
same limiting function, as the figures illustrate and as we will show below. 
We can not push numerical experiments far, as the number of vertices of the refinements grows so fast. 
But the theory confirms this. The proof is rather elementary using a Lidskii lemma. Here is the
main result: 

\begin{thm}[Central limit of barycentric refinements]
The sequence $F_{G_m}(x)$ converges in $L^1([0,1])$ 
to a limiting distribution function $F_d(x)$ 
which only depends on the dimension $d$ of the largest 
complete subgraph of $G$.
\end{thm}

The convergence is uniform on compact subsets in $[0,1)$ and exponentially fast in $L^1([0,1])$ norm
or in the $L^{\infty}([a,b])$ norm for every compact interval $[a,b] \subset (0,1)$. 
In the case $d=1$, where we can compute the spectrum for circular graphs, we have an
explicit limiting function $F$.  We first show that the limit $F_{G_m}(x)$ exists for each 
graph and then prove that the limit is {\bf universal} and only depends on the dimension $d$ of
the largest complete subgraph of $G$. 
The fact that a sequence of monotone functions $f_m$ in $L^1([0,1])$ which converges in $L^1([0,1])$ 
converges also uniformly on every closed interval $[a,b]$ with $0<a<b<1$ is known in real 
analysis (e.g. \cite{LewisShisha}). That the convergence can not be pushed 
to the boundary point $x=1$ follows from: 

\begin{lemma}
For $d \geq 2$, the function $F_d$ is not in $L^{\infty}([0,1])$, as the 
values $F_{G_m}(1)$ grow exponentially with $m$.
\end{lemma}

\begin{proof}
The Courant-Fischer estimate (which is a special case of the Schur inquality in linear algebra)
shows that $\lambda_{n-1} \geq {\rm max}_x({\rm deg}(x))$, if $n$ is the number of vertices.
Therefore, for every $x$: then ${\rm max} (v,Lv)/(v,v) = L_{xx} = d$. Since the maximal 
degree increases indefinitely, the maximal eigenvalue $\lambda_{n-1} = F_{G_m}(1)$ grows,
where $n=n(m)$ is the number of vertices of $G_m$. 
\end{proof} 

We do not know yet about the nature of the limiting functions in the case $d \geq 2$:

\question{
What is the spectral nature of the limiting density of states $\mu = F'$?
Does it have a discrete or singular part for $d \geq 2$? 
}

\section{The one dimensional case}

\begin{propo}
For $d=1$, the limiting function is $4 \sin^2(\pi x)$ so that the limiting
function $F$ is smooth. 
\end{propo}

\begin{proof}
The eigenvalues of the Laplacian of circular graphs $C_n$ are explicitly known using discrete
Fourier transform conjugating the Laplace operator $L$ on $l_2(Z_n)$ to the diagonal
matrix with entries $2-2 \cos(2 \pi k/n) = 4 \sin^2(\pi k/n)$ so that 
$$ \lambda_n = 4 \sin^2(\pi k/n)  \; . $$ 
By definition, we have  $F_{G_m}(x) \to 4 \sin^2(\pi x)$. 
\end{proof}

\begin{propo}
For $d=1$, the Laplacian of the circular graphs 
satisfies the renormalization map  $L(G_{m+1}) = \phi(T(L(G_m)))$,
where $T(x)=4x-x^2$ is the quadratic map with Julia set $[0,4]$ and 
$\phi$ restricts the matrix to an invariant subspace of half the dimension. 
\end{propo}
\begin{proof}
This is a direct verification for matrices. In the Fourier picture, it
becomes a double angle formula of trigonometry. Algebraically, it is an 
identity for Jacobi matrices, the Laplacian of the free particle on the circular
graph. For $G=C_4$ for example, we have the Laplacian
$$ L = \left[
                 \begin{array}{cccc}
                  2 & -1 & 0 & -1 \\
                  -1 & 2 & -1 & 0 \\
                  0 & -1 & 2 & -1 \\
                  -1 & 0 & -1 & 2 \\
                 \end{array}
                 \right] \; . $$
Now take the Laplacian $K$ of $C_8$ and form 
$T(K) = 4 K - K^2$. This gives
$$ K = \left[
                  \begin{array}{cccccccc}
                   2 & 0 & -1 & 0 & 0 & 0 & -1 & 0 \\
                   0 & 2 & 0 & -1 & 0 & 0 & 0 & -1 \\
                   -1 & 0 & 2 & 0 & -1 & 0 & 0 & 0 \\
                   0 & -1 & 0 & 2 & 0 & -1 & 0 & 0 \\
                   0 & 0 & -1 & 0 & 2 & 0 & -1 & 0 \\
                   0 & 0 & 0 & -1 & 0 & 2 & 0 & -1 \\
                   -1 & 0 & 0 & 0 & -1 & 0 & 2 & 0 \\
                   0 & -1 & 0 & 0 & 0 & -1 & 0 & 2 \\
                  \end{array}
                  \right] \; . $$
We see that that there are two $4$-dimensional subspaces
of $R^8$ on which $K$ is isomorphic to the old operator $L$. 
In the above notation, we have expressed this as $\phi(K) = L$. 
\end{proof} 

{\bf Remark:}
The quadratic map $T$ is associated to the Julia set of the map $x^2-2$ which is the
"tip of the tail" parameter $z=-2+i \cdot 0$ in the Mandelbrot set $M$. 
The polynomial $T$ is conjugated to a Tschebychev polynomial, which by the way are the polynomials
for which the Julia set is an interval \cite{Beardon}. Alternatively, in the context of
interval maps, the map $T$ is conjugated to the Ulam map $f(x) = 4x(1-x)$. 
The density of the eigenvalues follows the arcsin probability distribution 
$f(x) = (x(4-x))^{-1/2}/\pi$ supported on $[0,4]$.
It has with cumulative distribution function $(2/\pi) \arcsin(\sqrt{x}/2)$.
The absolutely continuous measure $\mu=f(x) 1_{[0,4]} dx$ is the natural equilibrium 
measure on the Julia set, the measure which maximizes metric entropy and equals it
to topological entropy $\log(2)$. 

\section{Convergence}

Intuitively, universal convergence is a consequence of the fact that
every simplex in $G_m$ spans more simplices producing a self similar pattern. 
Smaller dimensional parts and contributions from boundary get washed away as
the refinement progresses, as in dimension $d$, there are exponentially more 
points than in dimension $d-1$. Also in higher dimensions, we expect this to be 
a setup for a renormalization scheme, as the spectrum on part of the graph is 
close to the spectrum of the entire graph and different regions grow in the same manner.  \\

We can not write the limiting $F$ as a fixed point of a renormalization map yet. It can 
not be ruled out yet that the scalar Laplacian of $G_{m+1}$ could be related with the form 
Laplacian of $G_m$. In two dimensions, where only $1$-forms are Fermionic, 
the spectrum of $D^2$ is determined by the spectrum of the scalar $0$-forms and the
$2$-forms. \\

The refinements become more and more self-similar, even so the degrees become larger and larger. 
We will see that the average eigenvalue $\overline{\lambda}(G) = \sum_{k=1}^n \lambda_k/n = ||F_{G}||_1$
converges for $m \to \infty$ in such a way that $\overline{\lambda}(G_{m+1}) - \overline{\lambda}(G_m)$
decreases with an exponential rate depending only on the dimension of the largest complete subgraph. 

\question{Is there also in dimension $d>1$ a functional equation for which the 
limiting $F$ is a solution? Is there a relation with an equilibrium measure of a Julia 
set as in $d=1$?  }

When we look at convergence of $F_{G_m}$ in $L^1([0,1])$, the limiting 
graph density $2 v_1/v_0$ relating the number of edges $v_1$ with the number
of vertices $v_0$ matters. 

\begin{lemma}
$||F_{G_m}||_1 = \int_0^1 F_{G_m}(x) \; dx = 2 v_1(G_m)/v_0(G_m)$
which is the average vertex degree of $G_m$. 
\end{lemma}

\begin{proof}
The trace of the matrix $L$ is the sum of the eigenvalues. By the Euler handshaking lemma,
it is twice the number $v_1(G_m)$ of edges. The average eigenvalue is therefore 
$$   2v_1(G_m)/v_0(G_m) \; . $$
\end{proof}

In the case of a circular graph, we have $||F_{G_m}||+1=2$. 
In the case of the triangle graph, we have 
$$  ||F_{G_m}||_1 =  \frac{3\ 2^{n+1} \left(3^n+1\right)}{3\ 2^n+6^n+2}  \to 6   \; . $$
The fact that the average degree goes to $6$ follows of also from Gauss-Bonnet, as $K(x)=1-d(x)/6$
is the curvature for graphs without $K_4$ subgraphs.
Since the Euler characteristic of the triangle is $1$, the sum of the curvature has to 
converge to $1$. More generally, we have:

\begin{coro}
In any dimension $d$, the one has $||F_{G_m}||_1 \to (d+1)!$ exponentially fast.
\end{coro}

\begin{proof}
Indeed, as the number of boundary simplices grows like $(d!)^n$ and the number
of interior simplices grows like $((d+1)!)^n$, the convergence is of
the order  $1/(d+1)^n$.
\end{proof}

Since the largest vertex degree grows exponentially in $m$, we know $||F_{G_m}||_{\infty} \to \infty$.
The Cheeger inequality for graphs allows to say something about 
the ground state energy, the smallest nonzero eigenvalue of the connected graph $G$. 

\begin{coro}
The ground state energy $\lambda_1(G_m)$ converge to $0$ exponentially fast for $m \to \infty$. 
\end{coro}
\begin{proof}
In dimension $1$, where we know the explicit spectrum and because the number
of vertices grows exponentially. In general, it follows from the fact that the 
Cheeger number $|C(f)|/{\rm min}(|A(f)| |B(f)|)$ goes to zero exponentially fast.
\end{proof}

In order to compare eigenvalues, one could use a perturbation result of Weyl 
which tells that for selfadjoint matrices $A,B$, the eigenvalues $\lambda_k(A+B)$ are sandwiched
between $\lambda_k(A)+\lambda_1(B)$ and $\lambda_k(A)+\lambda_n(B)$. Since we need not only the 
individual eigenvalues to converge but also need a $l_1$-convergence,
the following variant of the Lidskii's theorem which was used already
in \cite{knillmckeansinger} is better suited:

\begin{lemma}[Lidskii]
For any two selfadjoint complex $n \times n$ matrices $A,B$ with
eigenvalues $\alpha_1 \leq \alpha_2 \leq \dots \leq \alpha_n$ and
$\beta_1 \leq \beta_2 \leq \dots \leq \beta_n$, one has
$\sum_{j=1}^n |\alpha_j - \beta_j| \leq \sum_{i,j=1}^{n} |A-B|_{ij}$.
\label{Lidskii}
\end{lemma}

The reduction to the standard Lidksii theorem was told me by Yoram Last
\cite{Last1995}:

\begin{proof}
Denote with $\gamma_i \in R$ the eigenvalues of the selfadjoint
matrix $C:=A-B$ and let $U$ be the unitary matrix diagonalizing $C$
so that ${\rm Diag}(\gamma_1, \dots ,\gamma_n)=UCU^*$. We calculate
\begin{eqnarray*}
 \sum_i |\gamma_i| &=& \sum_i (-1)^{m_i} \gamma_i
                    = \sum_{i,k,l} (-1)^{m_i} U_{ik} C_{kl} U_{il} \\
                    &\leq& \sum_{k,l} |C_{kl}| \cdot
                                    |\sum_i (-1)^{m_i} U_{ik} U_{il}|
                    \leq \sum_{k,l} |C_{kl}| \; .
\end{eqnarray*}
The claim follows now from Lidskii's inequality
$\sum_j |\alpha_j-\beta_j| \leq \sum_j |\gamma_j|$
(see \cite{SimonTrace})
\end{proof}

\begin{propo}
For any $d \geq 1$ and $G=K_{d+1}$, there exists a limiting eigenvalue distribution 
$F_d(x) = \lim_{m \to \infty} F_{G_m}(x)$. 
\end{propo}

\begin{proof}
We proceed by induction. For $d=1$, we have explicit eigenvalues.  For $d=2$, we 
make a refinement and divide it up into 6 pieces. Lidskii shows that the eigenvalues of $G_{m+1}$ 
consist of $6$ copies of the eigenvalues of $G_m$ plus a correction term which comes from lower 
dimensional interfaces and is much smaller. 
In words, $||F_{G_{m+1,2}}|| = ||F_{G_{m,2}}|| + 6 \cdot 3 ||F_{G_{m,1}}|| (2^n/6^n)$.
As $||F_{G_{m,1}}||$ converges, the sequence $||F_{G_{m,2}}||$ is a Cauchy sequence.
Now, lets go to the case $d=3$. We have
$||F_{G_{m+1,3}}|| = ||F_{G_{m,3}}|| + 24 \cdot 3 ||F_{G_{m,2}}|| (6^n/6^{n^2})$
and again have a Cauchy sequence showing that the limit exists in $L^1([0,1])$. 
In each dimension we get a limiting function as the lower dimensional interfaces between
the similar chambers grow exponentially slower than the chambers themselves.
\end{proof}

\begin{coro}
There exists a limiting density of states $\mu_d = F'_d(x)$ which is a measure on $[0,1]$. 
\end{coro}
\begin{proof}
As $F_{G_m}$ are monotone, it defines a measure $\mu_m$. Since the limiting 
eigenvalue distribution is monotone, it defines a measure $\mu$. 
\end{proof} 

With pointwise convergence, we would have weak convergence of $\mu_m$ to $\mu$. 
But we don't know that yet. \\

Lets now look at the proof of the theorem: we have to show that the limiting function $F$ does not depend
on the initial graph. Having seen convergence for a simplex $K_{d+1}$, it follows for
a finite union of $k$ simplices glued along lower dimensional simplices. Each simplex
evolves in the same way under the barycentric evolution and the average is the same 
function $F$. When cutting the graph apart or disregarding lower dimensional parts,
the modifications lead to change on a set of vertices which becomes exponentially 
less relevant as $m$ grows as it is lower dimensional. 
This is justified by Lidskii's estimate as we can estimate the sum of the eigenvalue differences. \\

\begin{figure}
\scalebox{0.1}{\includegraphics{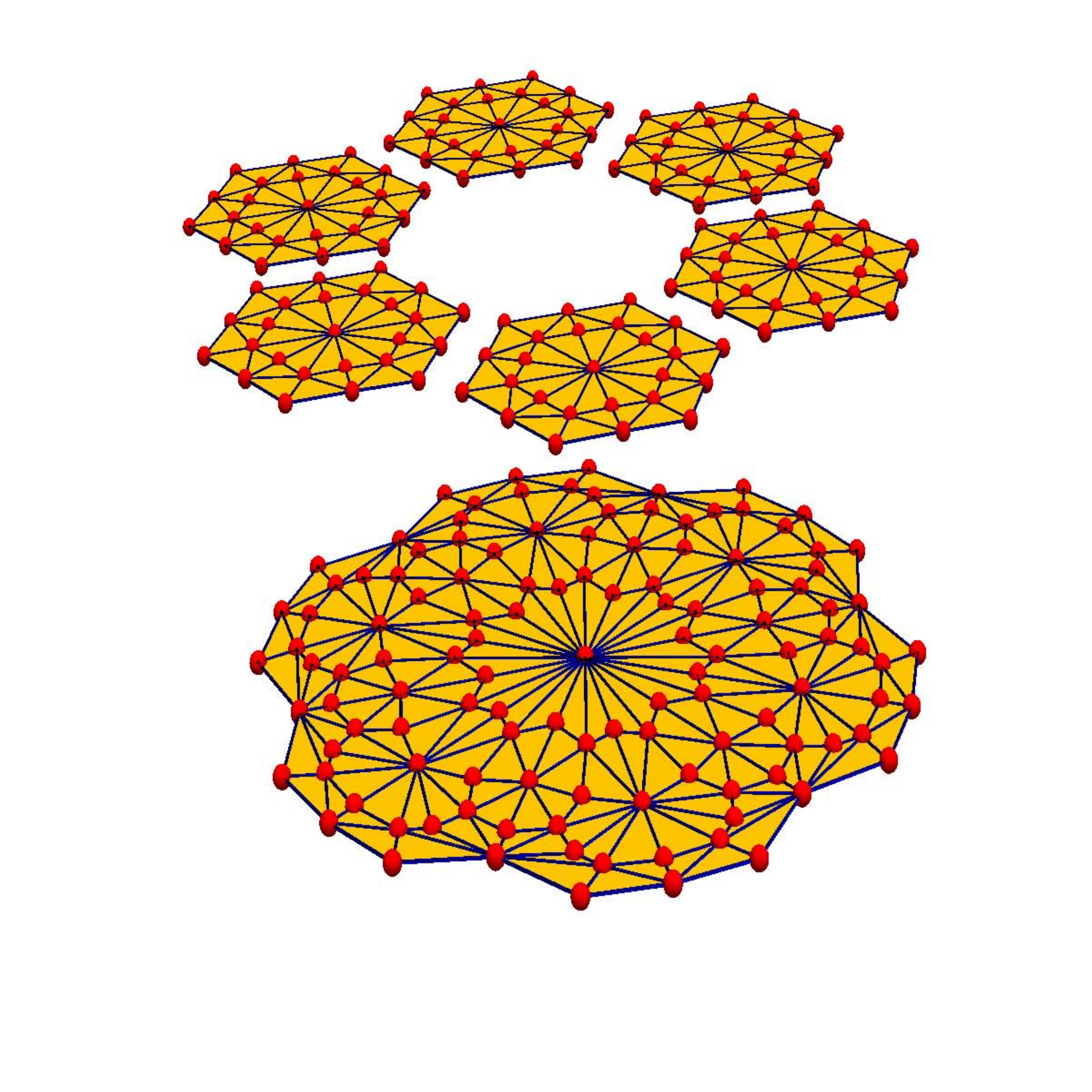}}
\caption{
For the proof, we cut the d-dimensional graph $G_m$ into $(d+1)!$ smaller pieces
of the form $G_{m-1}$. 
The modification on the boundary affects exponentially less vertices so that
by Lidskii, the $l_1$ difference between the spectra goes to zero exponentially
fast with $m$. This applies then to the $L^1([0,1])$ functions $F_{G_m}$. 
}
\end{figure}

\section{Vertex degree distribution}

Lets look at the vertex degree distribution $d_1 \leq d_2 \leq d_3 \dots $. In the 
same way as for the spectrum, we can define a degree distribution function 
$$   H_{G_m}(x) = d_{[x n]} \; . $$
The degree and eigenvalue distributions are linked by the Schur's inequality so that
one can deduce the following from the eigenvalues. It is also possible to do it 
directly: 

\begin{coro}[Limiting vertex distribution]
There exists a limiting vertex distribution function $H(x) = \lim_{m \to \infty} H_{G_m}(x)$. 
It is independent of the initial graph $G$ and depends only on the dimension $d$ of the largest
complete subgraph of $G$. The distribution 
satisfies $H(x) \geq F(x)$ for $x \in [0,1]$, $H(0)=F(0),H(1)=F(1)$.
\end{coro}
\begin{proof}
Splitting the graph $G_m$ into $(d+1)!$ smaller graphs $G_{m-1}$ and noting that the 
degrees of the lower dimensional walls can be neglected in the limit as the number
of vertices in those parts grows slower, we have a Cauchy sequence in $L^1$. \\
An inequality of Schur 
For any selfadjoint matrix $A$, the Schur inequality 
$$  \sum_{i=1}^t a_i \leq \sum_{i=1}^t  \mu_i $$ 
holds, where $\mu_i$ are the eigenvalues of $A$ and $a_i$ are the diagonal 
elements $a_i \geq a_{i+1}$ of $A$. In the case $t=n$ we have equality since
${\rm tr}(L) = \sum_i \lambda_i$. 
The Schur inequality gives $\sum_i^t d_i \geq \sum_{i=1}^t \lambda_i$
if the degrees $d_i$ of the graph and eigenvalues $\lambda_i$ of the graph 
are ordered in an ascending way.
\end{proof}

{\bf Remark.} A bit stronger than Schur is a result of Gone \cite{Brouwer}) which assures that
$$ \sum_{i=1}^k d_i \leq \sum_{i=1}^k \mu_i, 1 \leq k \leq n-1 $$
if $\mu_1 \geq \mu_2 \geq \dots \geq v_n=0$ are the eigenvalues of the Laplacian
of $G$ and $d_1 \geq d_2 \geq \dots \geq d_n>0$ are the vertex degrees of
the graph $G$. This implies that $H_{G_m}(x) > F_{G_m}(x)$ for each $m$ and suggest that 
also in the limit $H(x)> F(x)$ for every $x \in (0,1)$. \\

In the case $d=2$, the degree $d_i$ is related to the curvature $K(x) = 1- d(x)/6$
which adds up to the Euler characteristic of $G$. We see that the curvature
distribution is related to the eigenvalue distribution. \\

In general we can look at the curvature 
$$   K(x) = 1-\frac{V_0(x)}{2}+\frac{V_1(x)}{3}-\frac{V_2(x)}{4} \dots $$
of a vertex which by Gauss-Bonnet-Chern \cite{cherngaussbonnet} 
adds up to the Euler characteristic $\chi(G)$ of $G$.  \\

The limiting distribution works also for the Dirac operator $D=d+d^*$ and 
so for the Hodge Laplacian $D^2 = d d^* + d^* d$. The argument is the same. 

\begin{coro}
For any $k$, the Laplacian $L_k(G_m)$ on $k$-forms has a spectral limiting
function which only depends on $k$ and the dimension $d$ 
of the largest simplex in $G=G_0$.
\end{coro}

In each dimension $d$, the graph $G_m$ belonging to a simplex $G_0=K_d$ 
can be cut (using some modifications on smaller dimensional sub graphs) 
into $(d+1)!$ smaller isomorphic pieces, which are of the form 
$G_{m-1}$. The $(d-1)$-dimensional cuts do not matter in the limit as the
number of simplices in them grows exponentially less fast.  \\

\begin{figure}
\scalebox{0.1}{\includegraphics{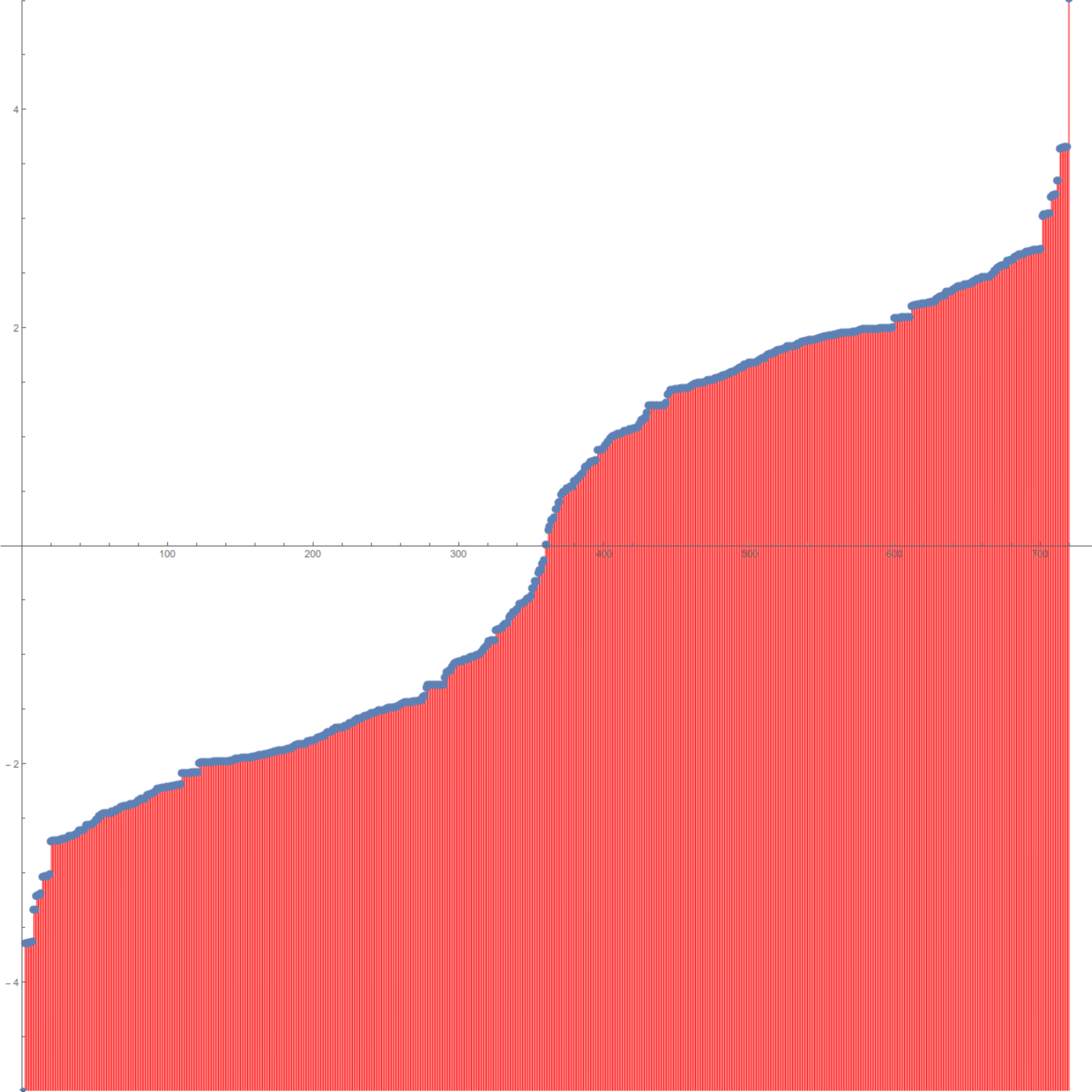}}
\scalebox{0.1}{\includegraphics{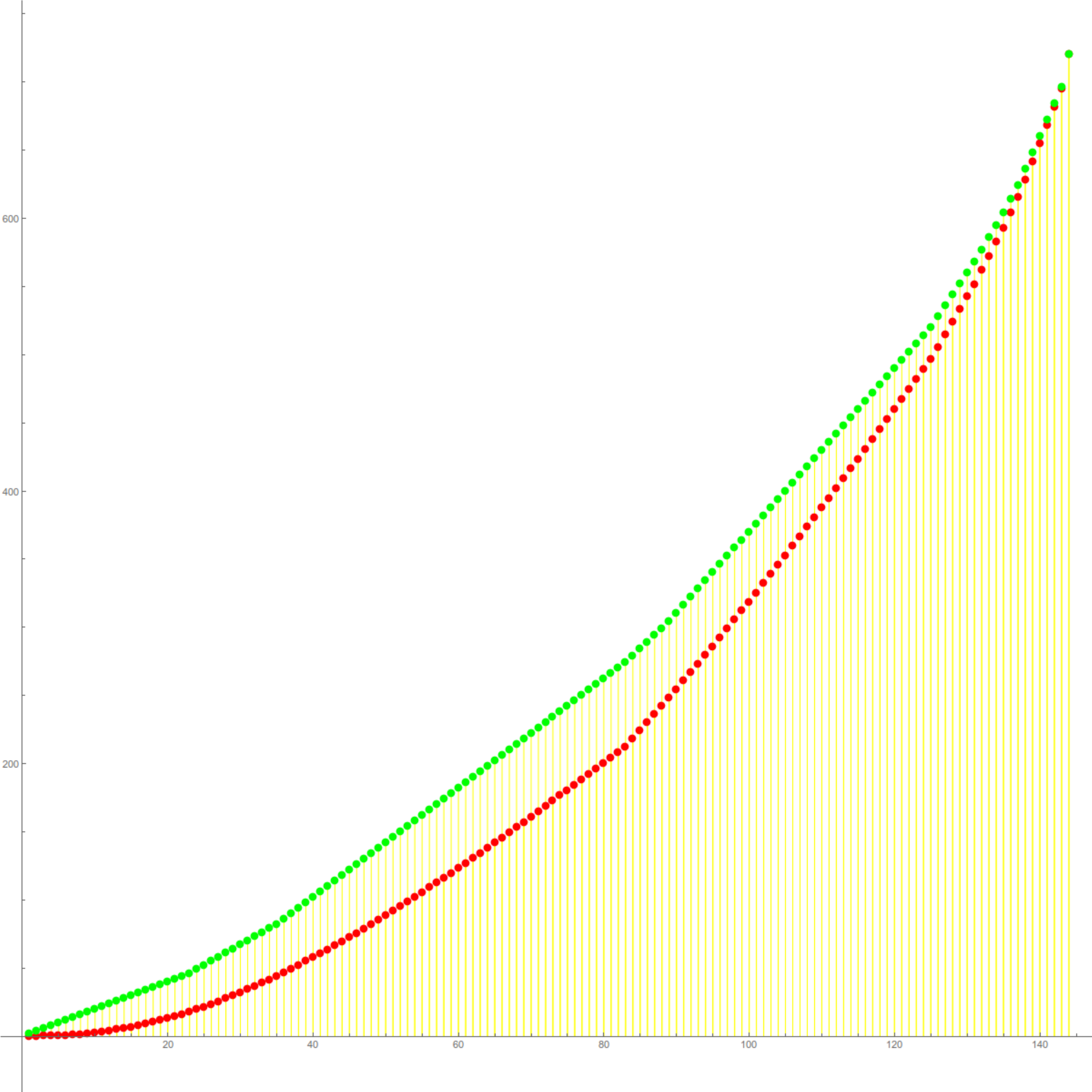}}
\caption{
To the left is a graph of the spectral density function 
of the Dirac operator $d+d^*$ of a barycentric refinement of the house graph. 
To the right an illustration of the Schur inequality: the 
integrated vertex degree function is an upper bound for the integrated 
eigenvalue function $\int_0^x F(t)\; dt$ and agrees at the end points $0,1$. }
\end{figure}

Since the matrices $D^2(G_m)$ and $L(G_{m+1})$ have the same size, we compared
the distribution of the Hodge Laplacian on the level $m$ with the distribution of
the scalar Laplacian on the level $m+1$ and they appear comparable. 
Because of super symmetry, the 1-form Laplacian $L_1$ has the same spectrum 
than the union of the $0$ and $2$-form Laplacians. The $2$-form Laplacian is
related to the Laplacian of the dual graph. \\

There might be more relations if things are extended to Schr\"odinger operators.
In \cite{Kni95}, we looked at renormalization
maps $L \to D$ satisfying $L=D^2+c$, where $D$ is a Laplacian on a barycentric
refined circular graph. Its not clear whether there are algebraic relations for
$d \geq 2$.  For a triangle $G=G_0$ with Laplacian
$\left[
                 \begin{array}{ccc}
                  2 & -1 & -1 \\
                  -1 & 2 & -1 \\
                  -1 & -1 & 2 \\
                 \end{array} \right]$ 
the Laplacian $L(G_1)$ of the refinement $G_1$ and the form Laplacian $D^2(G_0)$ 
have the same size
$$ L(G_1) = \left[
                 \begin{array}{ccccccc}
                  3 & -1 & 0 & -1 & -1 & 0 & 0 \\
                  -1 & 3 & -1 & -1 & 0 & 0 & 0 \\
                  0 & -1 & 3 & -1 & 0 & 0 & -1 \\
                  -1 & -1 & -1 & 6 & -1 & -1 & -1 \\
                  -1 & 0 & 0 & -1 & 3 & -1 & 0 \\
                  0 & 0 & 0 & -1 & -1 & 3 & -1 \\
                  0 & 0 & -1 & -1 & 0 & -1 & 3 \\
                 \end{array}
                 \right] \;, $$
$$  D^2(G_{0}) = \left[
                 \begin{array}{ccccccc}
                  2 & -1 & -1 & 0 & 0 & 0 & 0 \\
                  -1 & 2 & -1 & 0 & 0 & 0 & 0 \\
                  -1 & -1 & 2 & 0 & 0 & 0 & 0 \\
                  0 & 0 & 0 & 3 & 0 & 0 & 0 \\
                  0 & 0 & 0 & 0 & 3 & 0 & 0 \\
                  0 & 0 & 0 & 0 & 0 & 3 & 0 \\
                  0 & 0 & 0 & 0 & 0 & 0 & 3 \\
                 \end{array}
                 \right] \; .  $$

\section{Figures} 

\begin{figure}
\scalebox{0.08}{\includegraphics{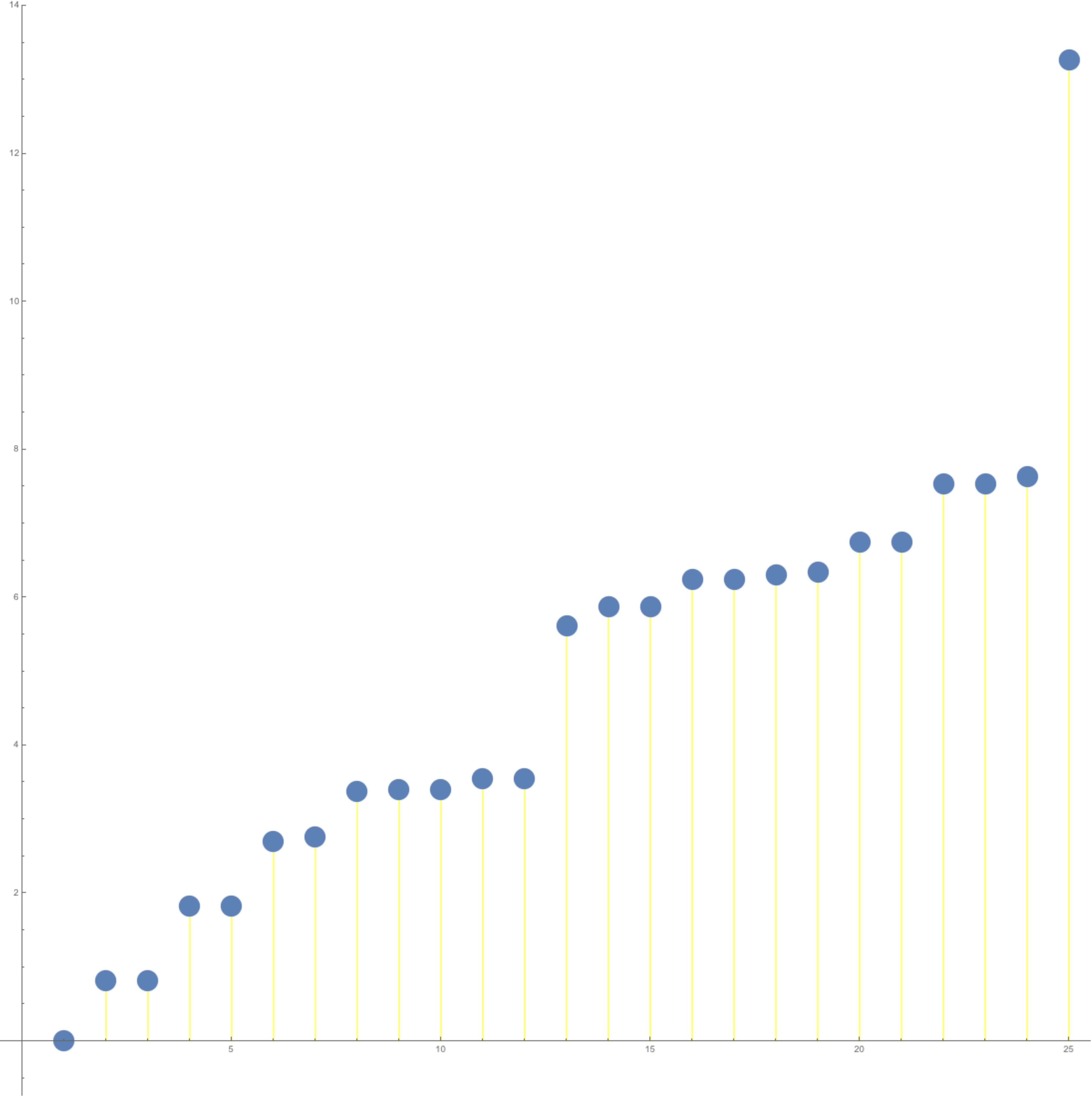}}
\scalebox{0.08}{\includegraphics{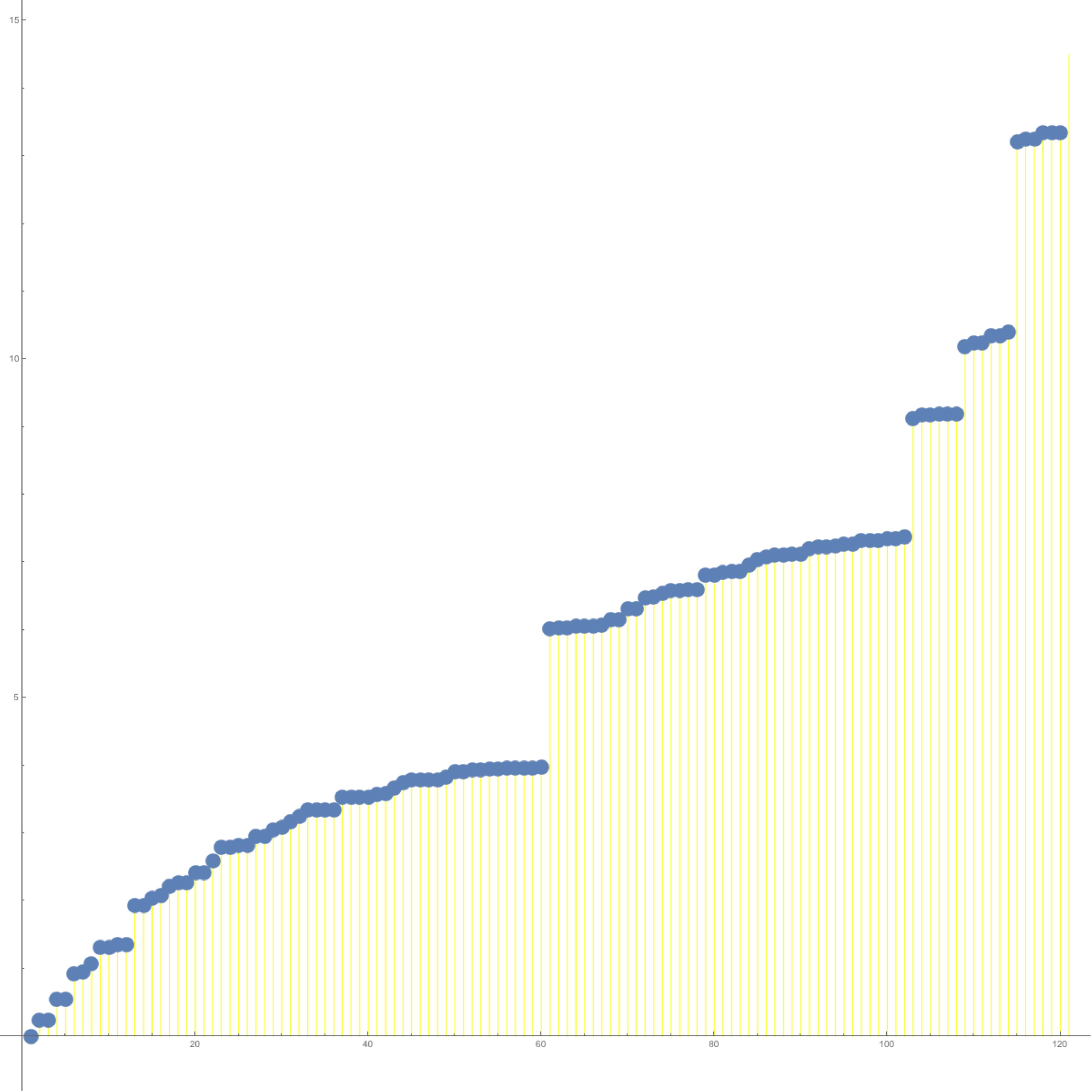}}
\scalebox{0.08}{\includegraphics{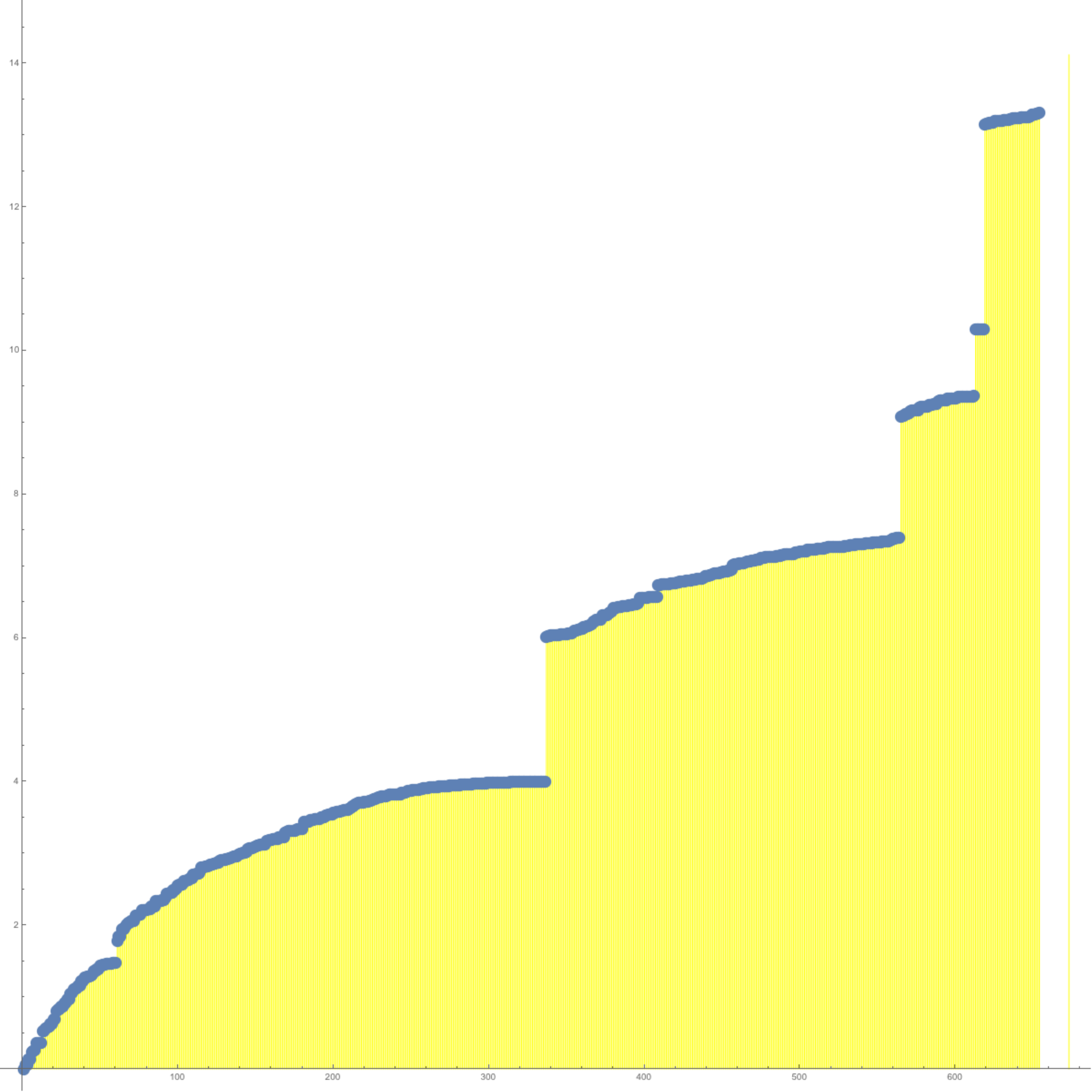}}
\caption{
The spectrum of the barycentric refinements $G_m$ of the triangle,
for $m=3,4,5$.
}
\end{figure}

\begin{figure}
\scalebox{0.08}{\includegraphics{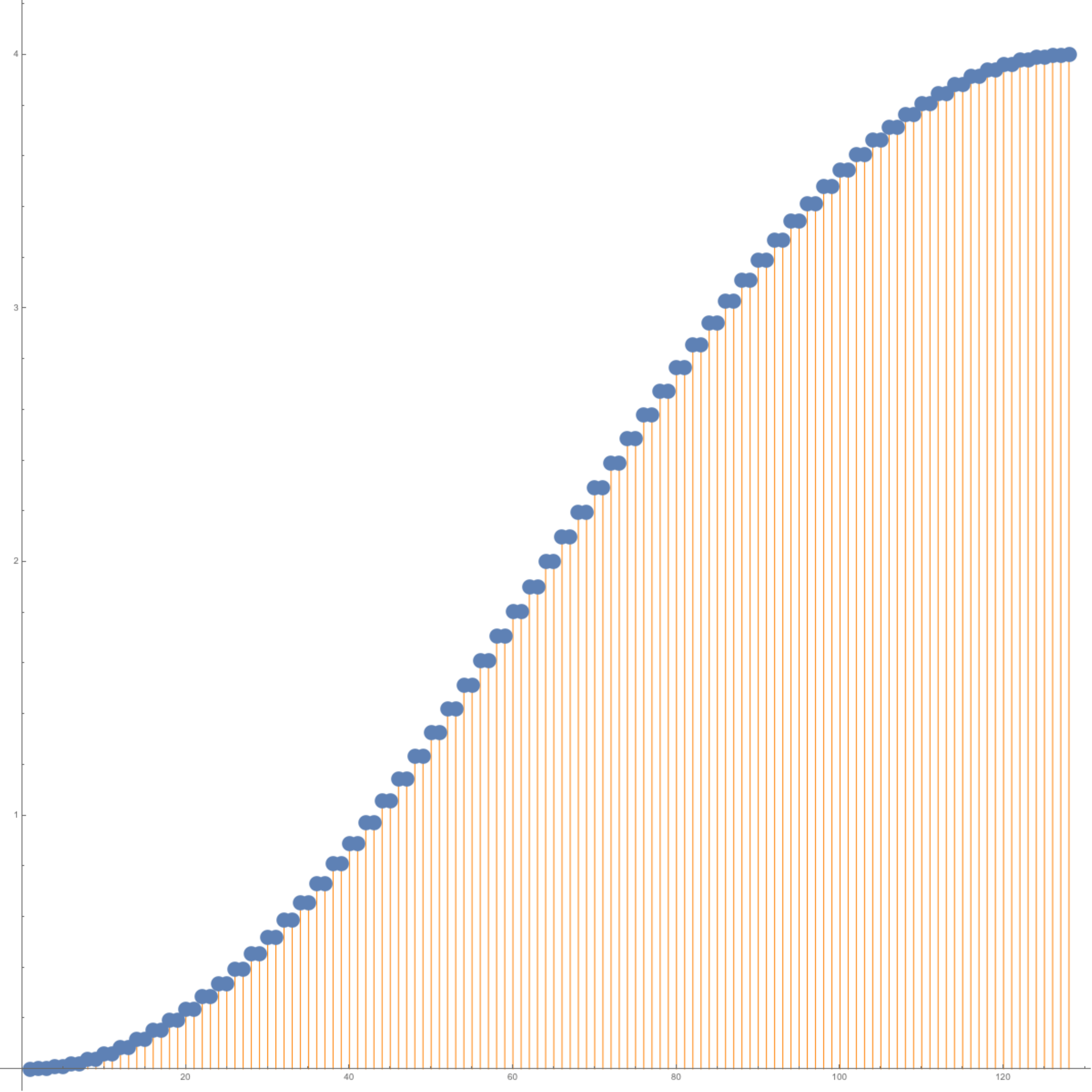}}
\scalebox{0.08}{\includegraphics{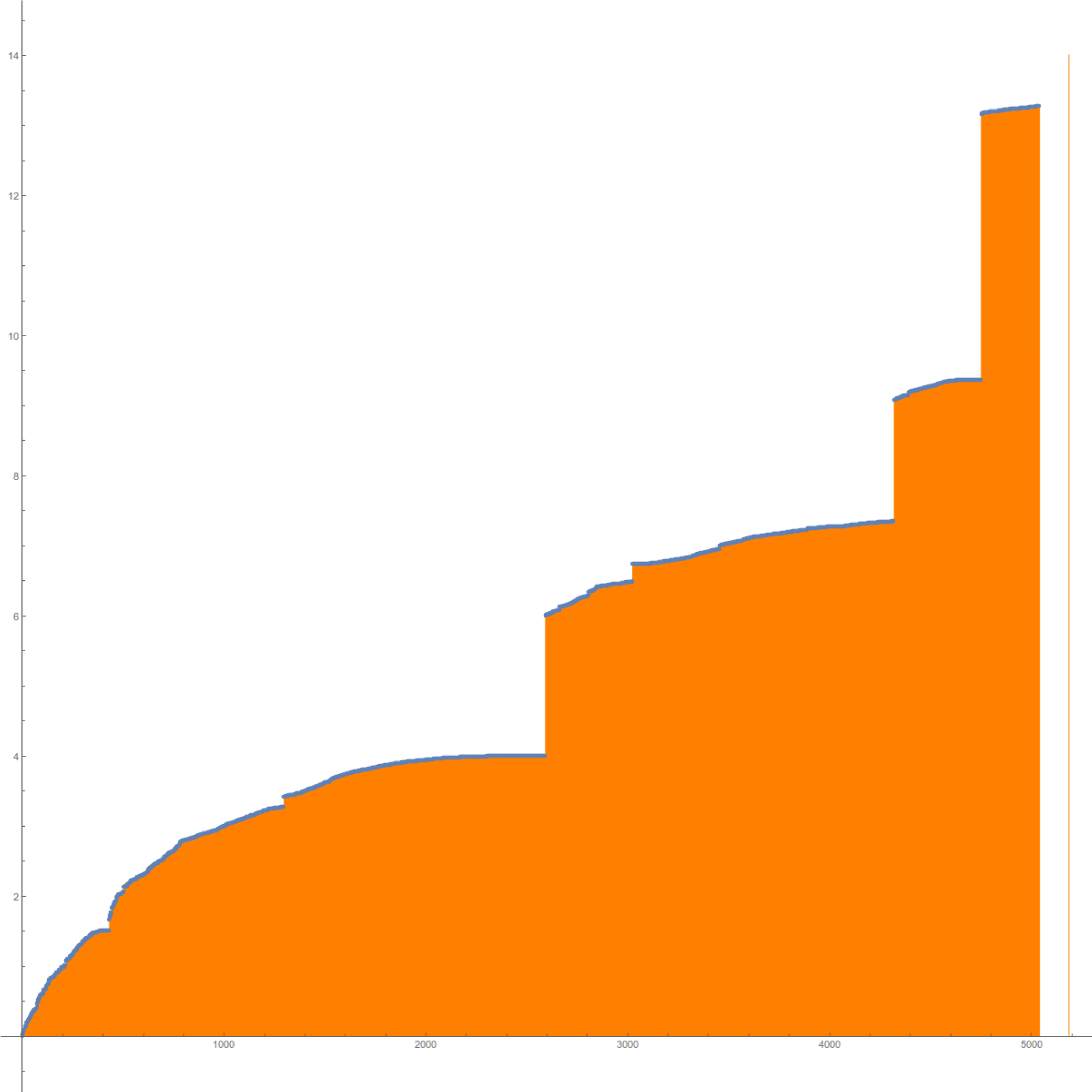}}
\scalebox{0.08}{\includegraphics{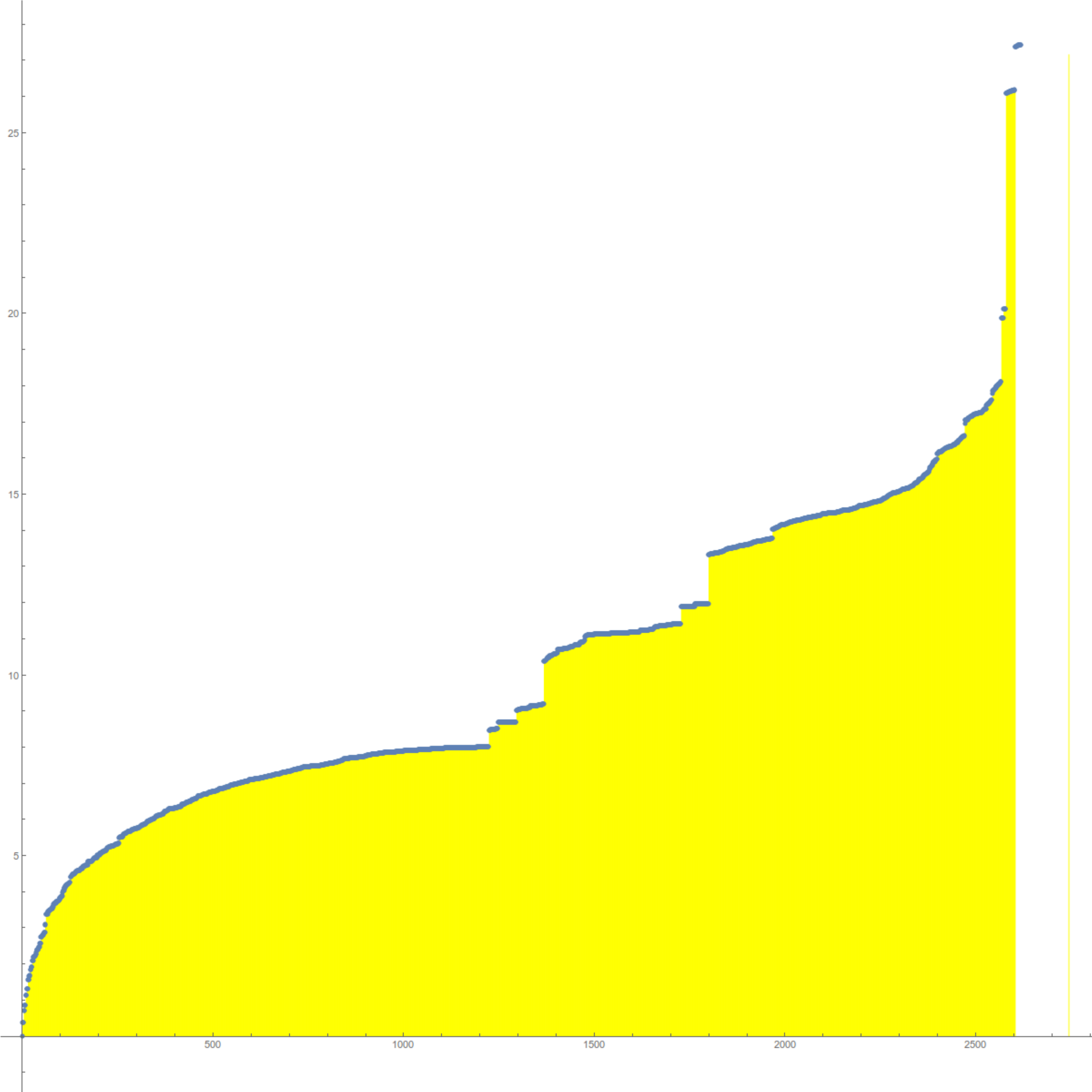}}
\caption{
The spectral functions $F$ in the case $d=1,d=2$ and $d=3$. 
We see the spectrum of barycentric refinements $G_7$ for $G=C_4$,
the spectrum of the refinement $G_4$ of the octahedron $G$ 
and finally the spectrum the refinement $G_3$ of the 
tetrahedron $K_4$. 
}
\end{figure}

\begin{figure}
\scalebox{0.08}{\includegraphics{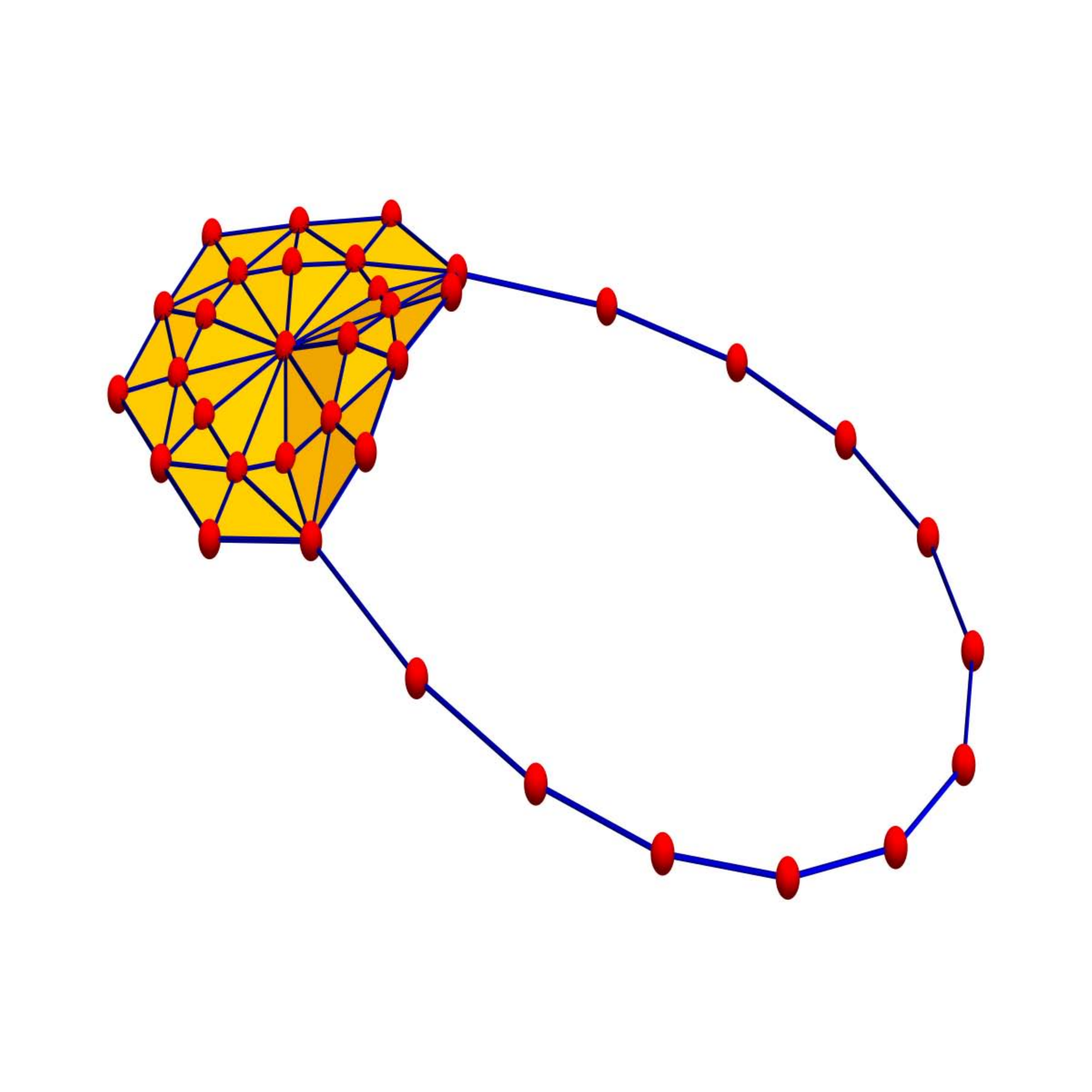}}
\scalebox{0.08}{\includegraphics{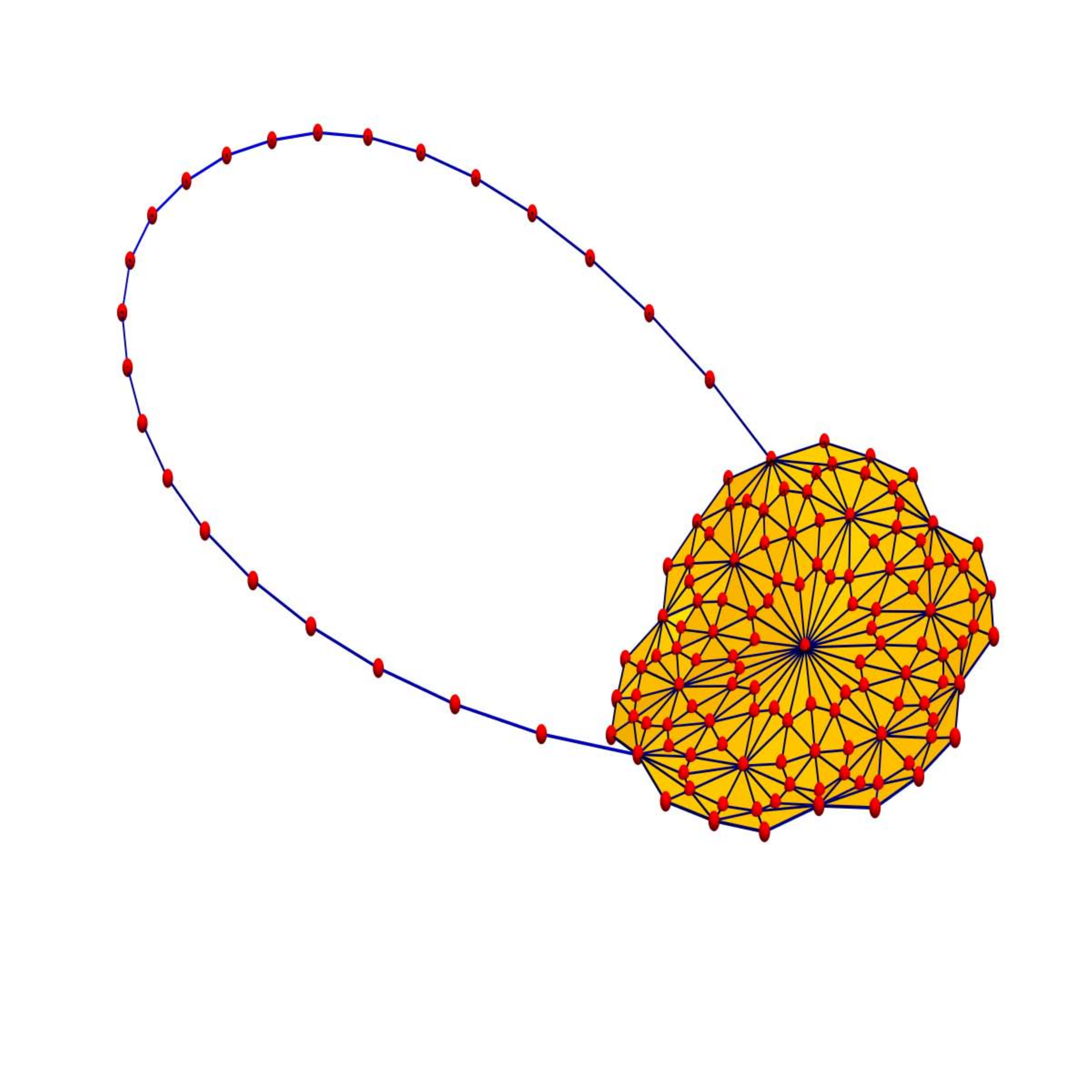}}
\scalebox{0.08}{\includegraphics{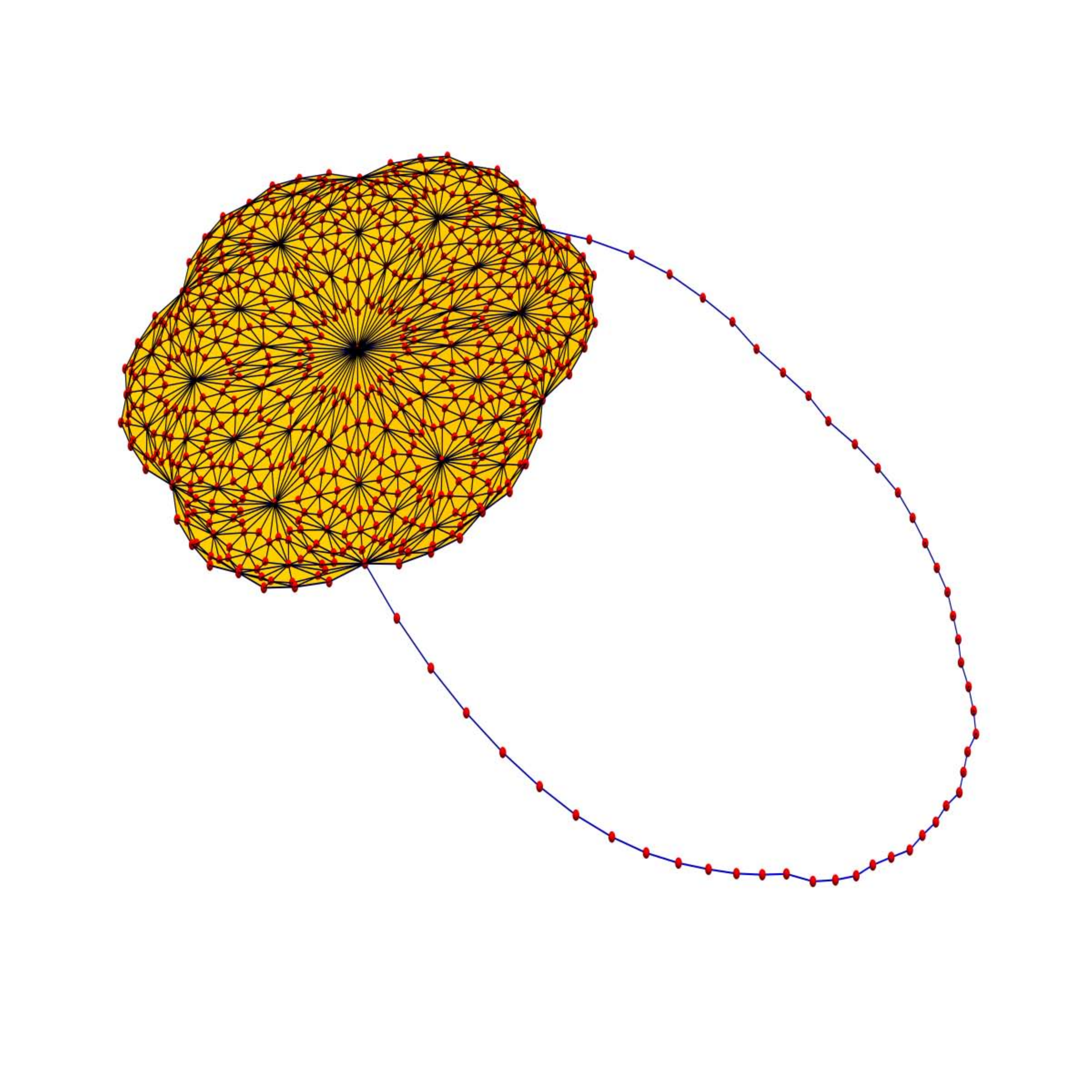}}
\caption{
Barycentric refinements of the house graph.
The one dimensional component can be neglected in the limit.
}
\end{figure}

\begin{figure}
\scalebox{0.08}{\includegraphics{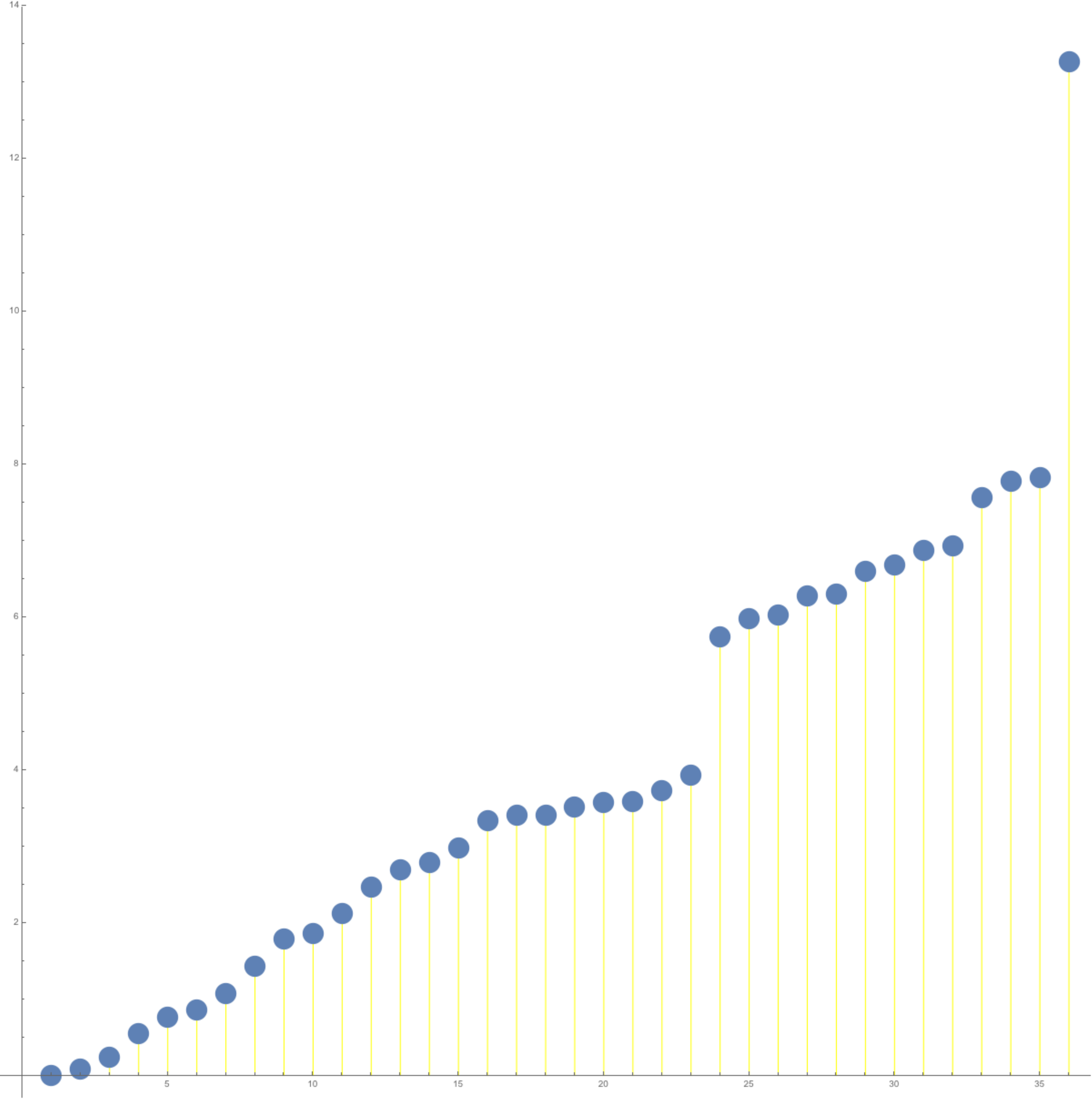}}
\scalebox{0.08}{\includegraphics{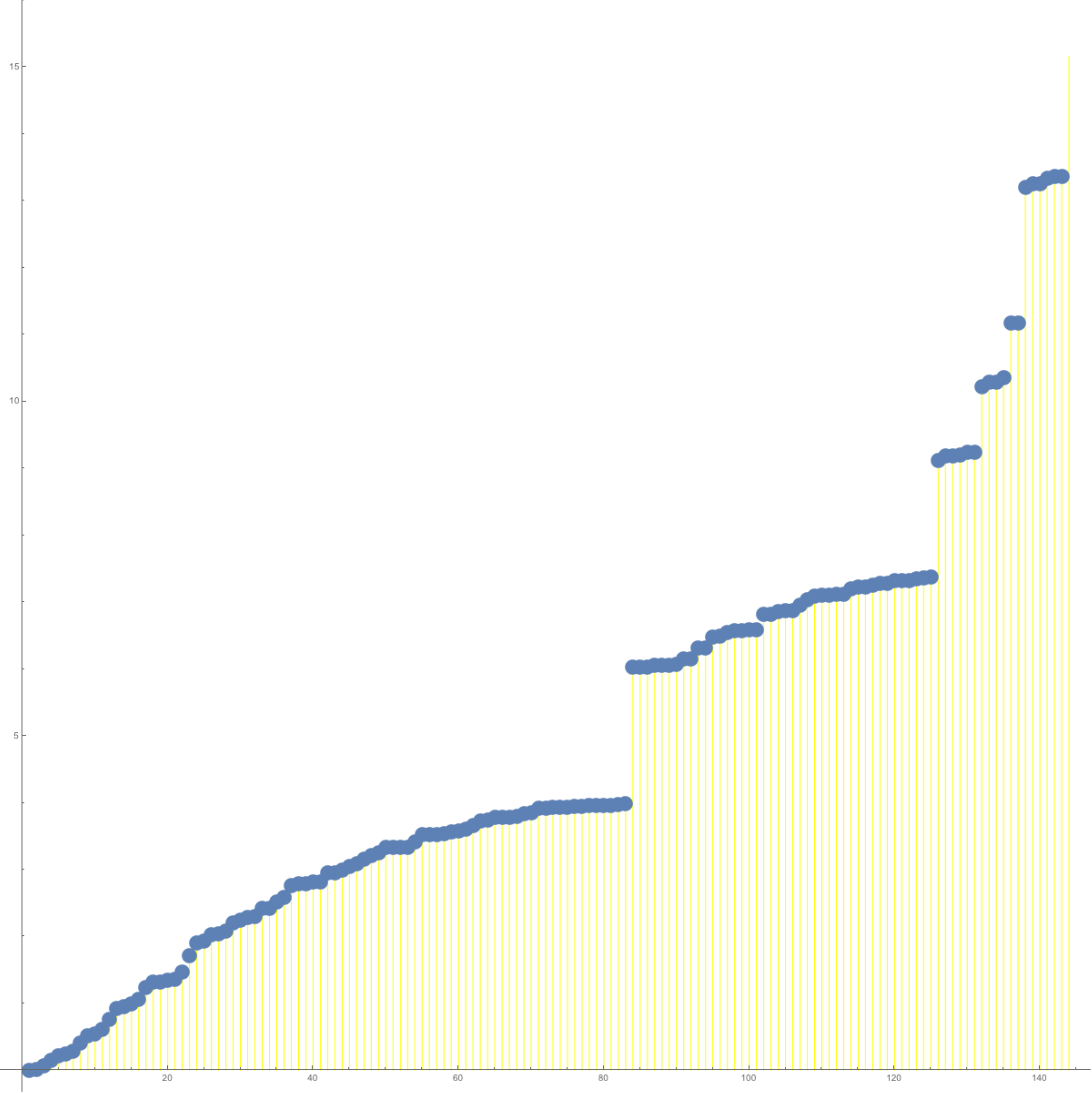}}
\scalebox{0.08}{\includegraphics{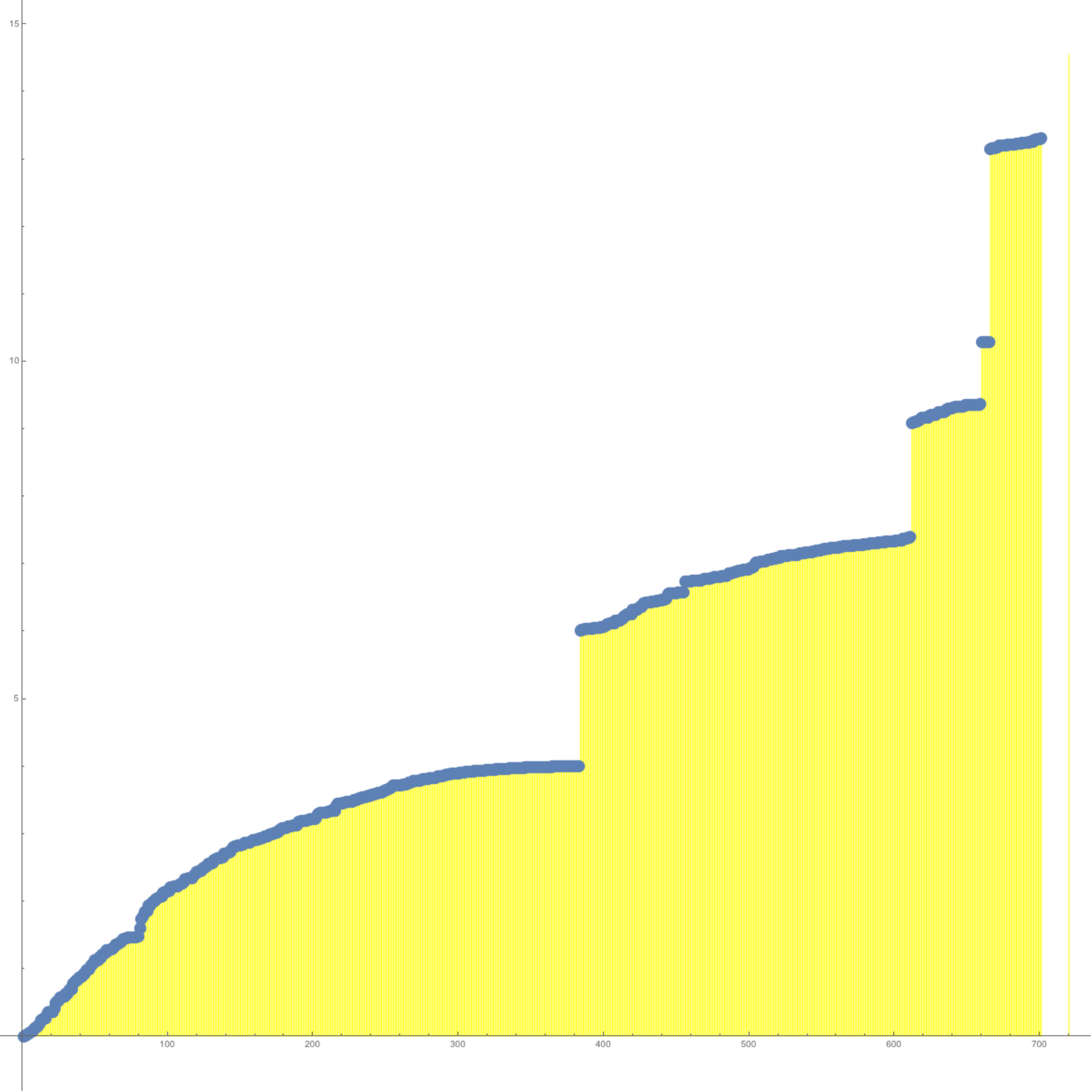}}
\caption{
The spectrum of barycentric refinements $G_m$ of the house graphs
for $m=3,4,5$, which is the case $d=2$. 
}
\end{figure}

\begin{figure}
\scalebox{0.08}{\includegraphics{figures/triangle5.pdf}}
\scalebox{0.08}{\includegraphics{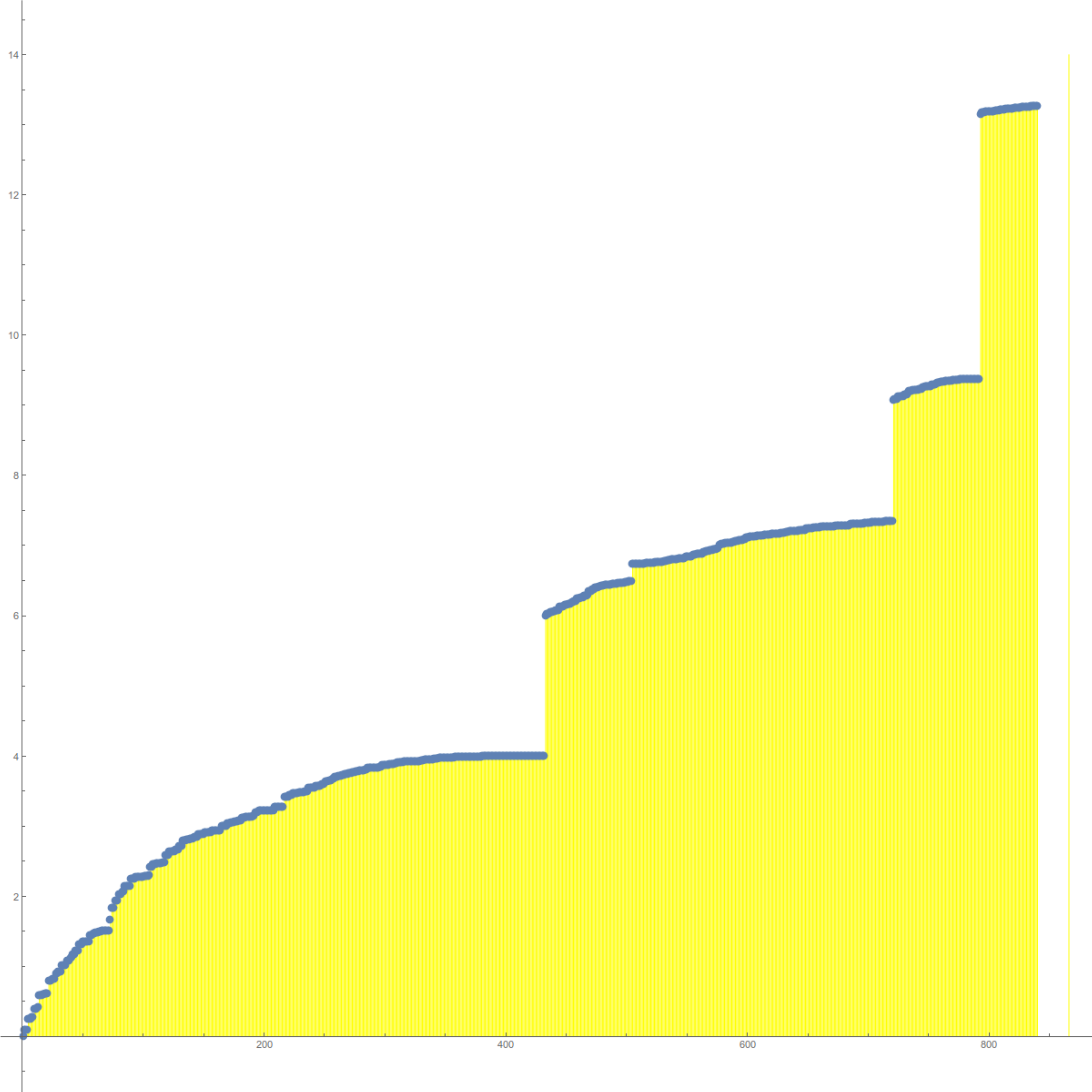}}
\scalebox{0.08}{\includegraphics{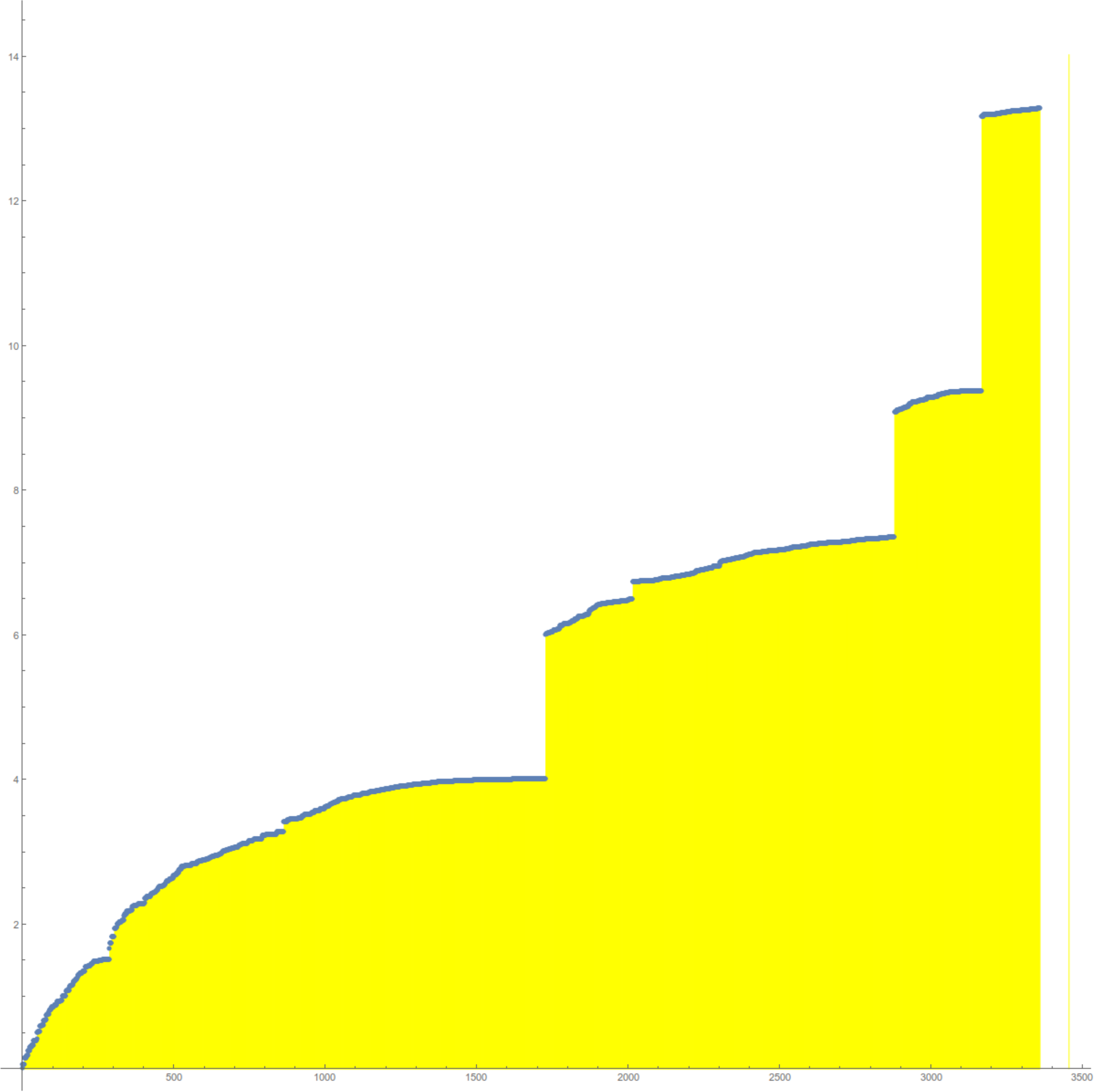}}
\caption{
The spectrum of barycentric refinements $G_4$ for the triangle
the octahedron and the torus. The eigenvalue distribution 
look almost identical. These are all 2 dimensional cases.
}
\end{figure}

\begin{figure}
\scalebox{0.1}{\includegraphics{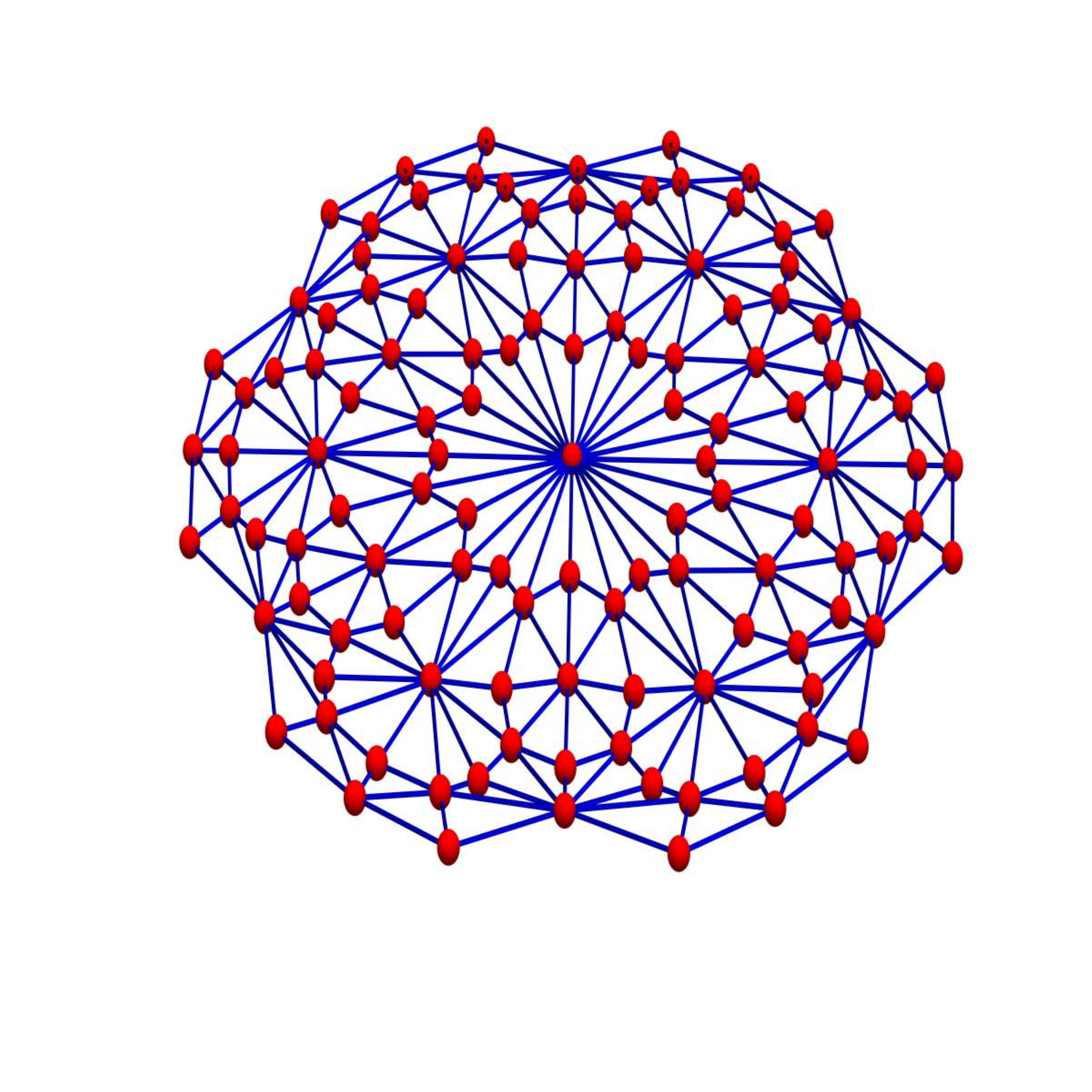}}
\scalebox{0.1}{\includegraphics{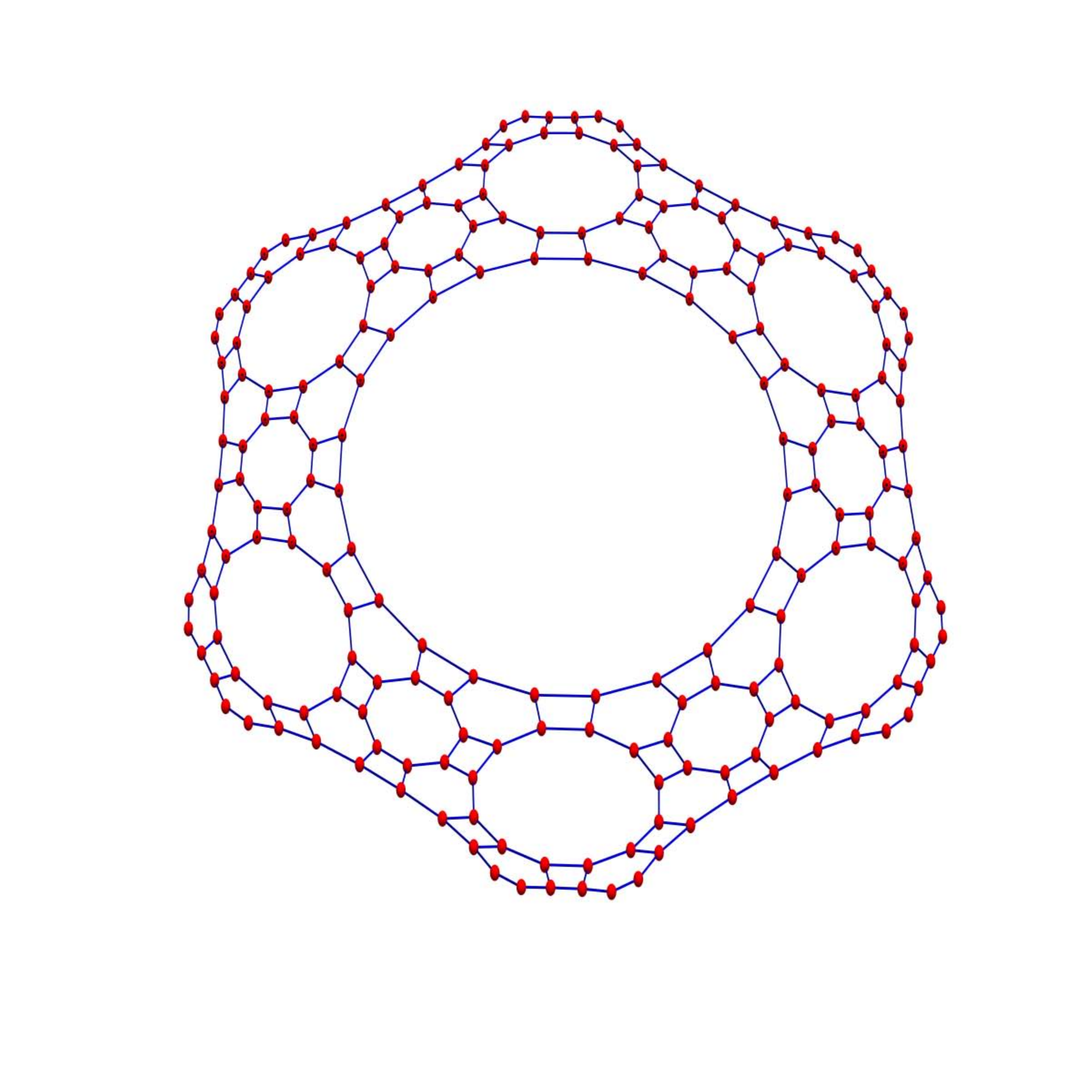}}
\caption{
A refinement of the triangle and its dual graph: the graph 
whose vertices are triangles and where two triangles are connected
if they intersect in an edge. The dual graph does not have
triangles.
}
\end{figure}

\begin{figure}
\scalebox{0.1}{\includegraphics{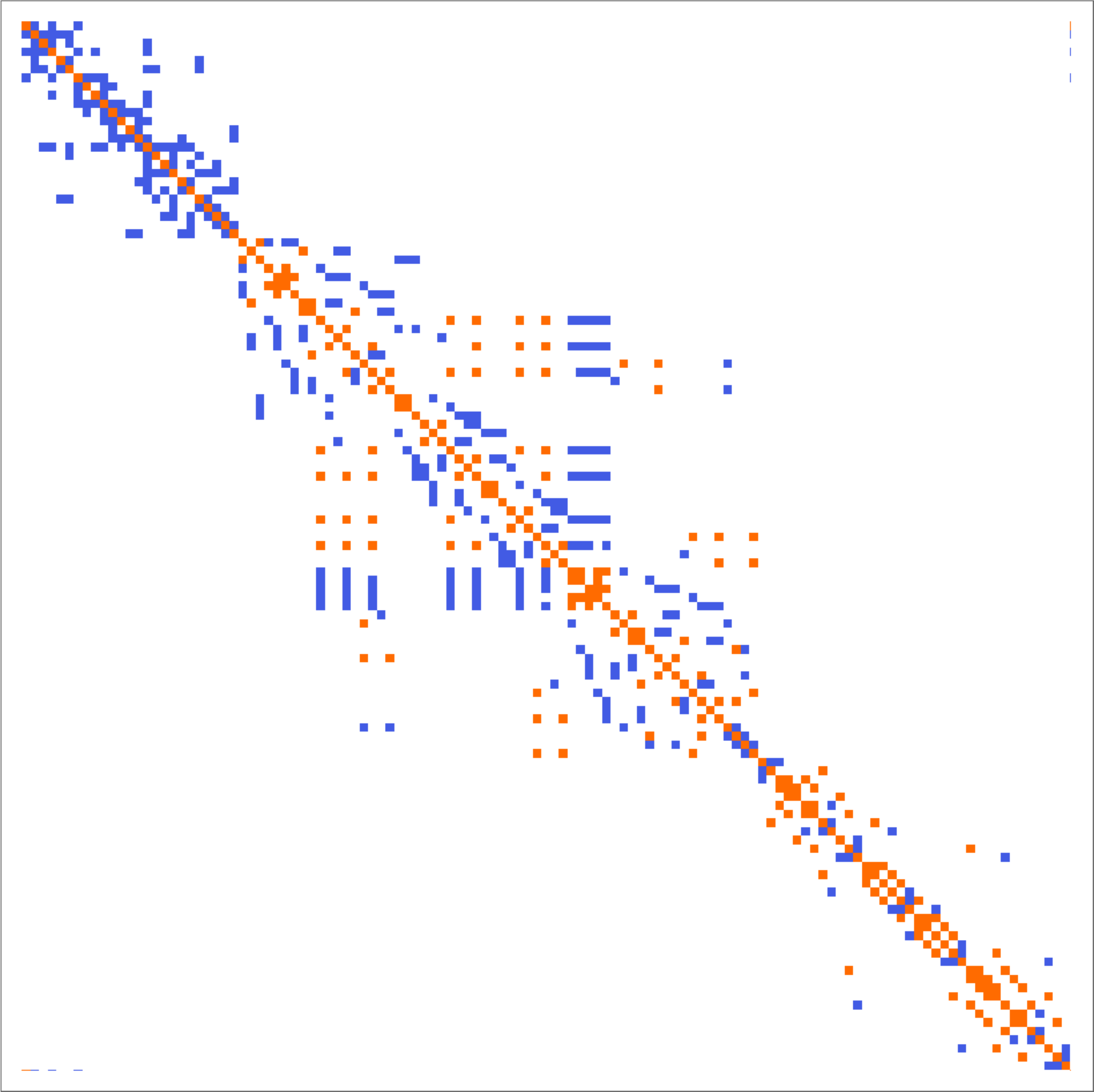}}
\scalebox{0.1}{\includegraphics{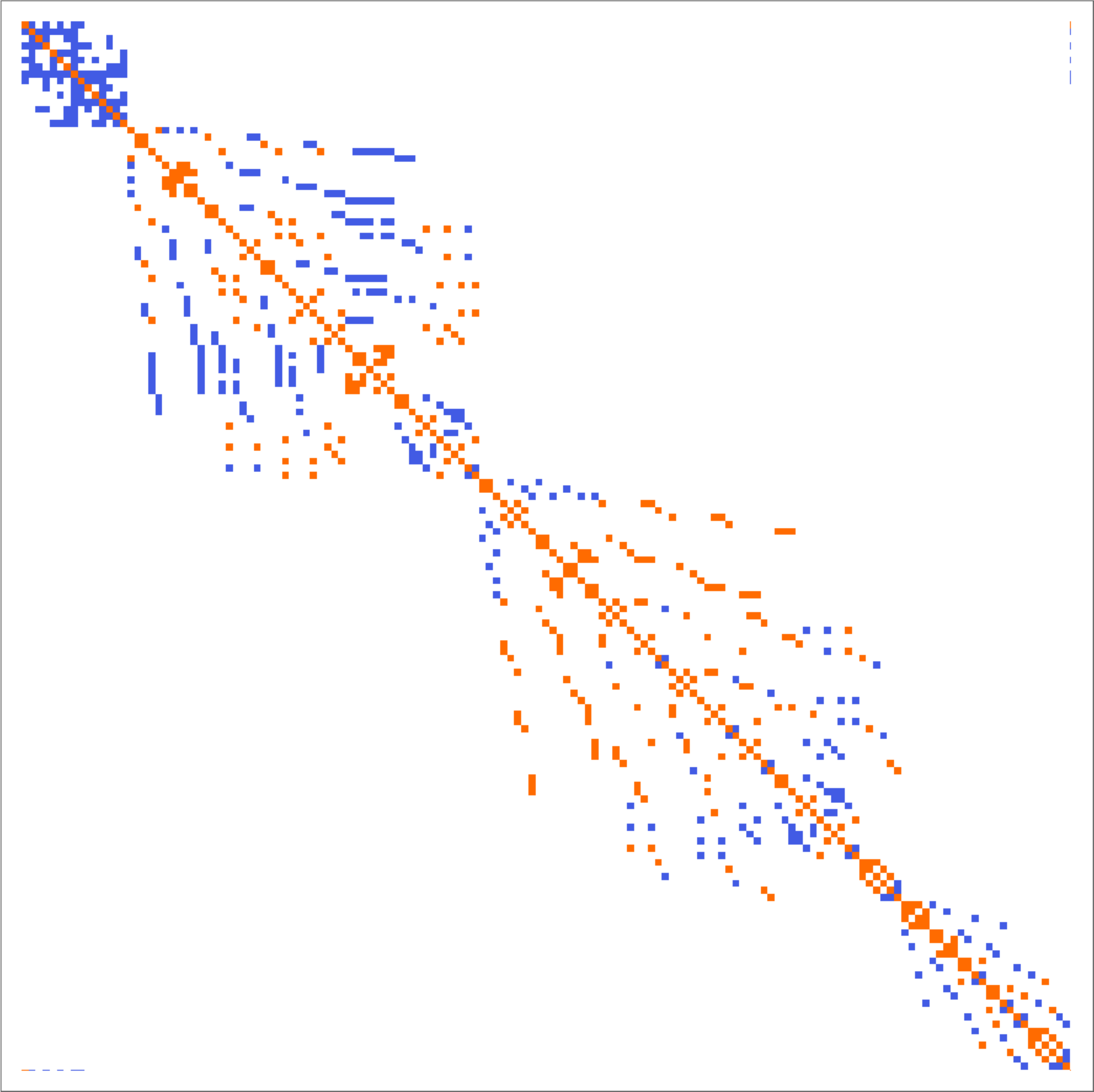}}
\caption{
The Hodge Laplacian $L=(d+d^*)^2$ for refinements
of the triangle $K_3$ has three blocks $L_0,L_1,L_2$, where $L_k$
acts on $k$-forms. The Hodge Laplacian of the tetrahedron $K_4$
has blocks $L_0,L_1,L_2,L_3$. Super-symmetry 
assures that the union of the spectra of the 
Bosonic parts $L_0,L_2$ is the union of the spectra of the Fermionic parts $L_1,L_3$. 
}
\end{figure}

\bibliographystyle{plain}

\end{document}